%% file: main.tex
\documentclass[12pt]{article}
\input{packages}

\numberwithin{equation}{section}
\allowdisplaybreaks

\title{Local Projections Inference with High-Dimensional Covariates without Exact Sparsity}
\author{
Jooyoung Cha\thanks{Jooyoung Cha: \texttt{jcha23@nd.edu}.
Department of Economics,
University of Notre Dame }
\thanks{I am grateful to Yuya Sasaki, Atsushi Inoue, Harold D. Chiang, Ke-Li Xu, and Vadim Marmer for their guidance. I also thank participants at NY Camp Econometrics, SETA 2024, IAAE 2024 (China and Greece), ESG Bristol, CESG 2024, and MEG 2024, special appreciation for the MEG Women’s Mentoring Session for their valuable comments and suggestions. I also thank seminar participants at Maastricht, Aarhus, UNC, Monash, LSE, Sciences Po, and Notre Dame. 
This paper was awarded the Best Student Paper Prize at IAAE 2024 Xiamen.
I am grateful to the editor and three anonymous referees for their careful reading and constructive suggestions, which substantially improved the paper.} }
\date{\vspace*{-.2cm}November 9, 2026}

\begin{document}
\maketitle
\vspace*{-1cm}
\begin{abstract}
    This paper develops a high-dimensional local projections framework for estimating impulse responses in the presence of many potential controls. Other approaches often rely on exact sparsity assumptions, but empirical macroeconomic applications may involve nuisance relationships better described by structured coefficient decay. Building on the orthogonal greedy algorithm with high-dimensional AIC (OGA+HDAIC), I propose a double-selection procedure that separately addresses the shock equation and the outcome equation within local projections. 
    The method allows for nonexact sparsity through polynomially decaying nuisance coefficients, with rates indexed by the decay exponent rather than by an effective sparsity level.
    Under suitable moment, dependence, and approximation conditions, the estimator achieves asymptotic normality with HAC standard errors. 
    Simulation evidence shows reliable coverage under shock contamination and lower point-estimation error under predictable high-dimensional outcome nuisance.
    I illustrate the method's practical feasibility in a monetary policy application using the FRED-MD dataset.
\end{abstract}

\clearpage
\section{Introduction}
\input{sections/1_Intro_v2}

\section{Method Overview}\label{sec:overview}
\input{sections/2_Overview}


\section{Theory}\label{sec:theory}
\input{sections/4_Theory}

\section{Simulation}\label{sec:simulation}
\input{sections/5_Simulation}

\section{Empirical Application}\label{sec:applications}

\input{sections/6_Empirical}

\section{Conclusion}\label{sec:conclusion}
\input{sections/7_Conclusion}

\clearpage
\bibliographystyle{ecta}
\bibliography{biblio}

\clearpage
\appendix
\input{sections/8_Appendix_v2.tex}\label{sec:appendix}

\end{document}

%% file: packages.tex
\usepackage{amssymb,amsmath,amsthm,amsfonts}
\usepackage{graphicx}
\usepackage{enumerate}
\usepackage{bm}
\usepackage{bbm} 
\usepackage{dsfont}
\usepackage{threeparttable}
\usepackage{subcaption}
\usepackage{fullpage,multirow, multicol,tasks,booktabs,tikz}
\usepackage[authoryear,longnamesfirst]{natbib}
\usepackage[shortlabels]{enumitem}
\usepackage[singlespacing]{setspace}
\usepackage{rotating}
\usepackage{array}
\usepackage{tabularx}

\definecolor{orangee}{RGB}{233, 116, 81}
\definecolor{green}{rgb}{0,0.5,0}
\definecolor{ndblue}{HTML}{0c2340}
\definecolor{ndgold}{HTML}{ae9142}
\definecolor{slateblue}{HTML}{8a804d}
\definecolor{iris}{HTML}{8a804d}

\usepackage{mathpazo}
\usepackage{microtype}

\usepackage{hyperref}
\hypersetup{
colorlinks=true,
linkcolor=iris,
filecolor=magenta,      
urlcolor=slateblue,
citecolor = slateblue
}
\usepackage{cleveref}

\usepackage{tcolorbox}
\newtcbox{\feedback}{nobeforeafter,colframe=black,colback=white,boxrule=0.5pt,arc=2pt,
  boxsep=0pt,left=2pt,right=2pt,top=2pt,bottom=2pt,tcbox raise base}
      
\usepackage{amsthm}
 
\newtheorem{thm}{Theorem}[section]
\newtheorem{asm}{Assumption}

\newtheorem{prop}{Proposition}[section]
\newtheorem*{prop*}{Proposition}

\newtheorem{cor}{Corollary}[section]
\theoremstyle{definition}
\newtheorem{rem}{Remark}
\newtheorem*{rem*}{Remark}

\newtheorem{algorithm}{Algorithm}

\newtheorem{lemma}{Lemma}
\newtheorem*{not*}{Notations}
\newtheorem*{notdef*}{Notations and Definitions}

\newcommand{\norm}[1]{\left\lVert#1\right\rVert}
\newcommand{\sumt}{\sum_{t=1}^T}

\newcommand{\abs}[1]{\left|#1\right|}
\newcommand{\paren}[1]{\left(#1\right)}
\DeclareRobustCommand{\frac}[3][0pt]{%
  {\begingroup\hspace{#1}#2\hspace{#1}\endgroup\over\hspace{#1}#3\hspace{#1}}}

\newcommand{\wt}{\boldsymbol{w}_t}

\newcommand{\bw}{\boldsymbol{w}}

%% file: sections/1_Intro_v2.tex
Econometricians use impulse responses to analyze how a shock, such as a change in interest rates, affects various aspects of the economy over time.
Knowing how different parts of the economy respond to shocks, such as policy changes or natural shocks, helps policymakers make informed decisions and businesses prepare for potential impacts.
In recent years, the local projections (LP) method, introduced by \cite{jorda2005estimation}, has gained considerable attention as a flexible alternative to traditional vector autoregression (VAR) models. LP directly estimates the relationship between shocks and future outcomes, offering simplicity and intuitive interpretability.

As the use of LP has grown, so has recognition of the need for high-dimensional controls in these models. 
While methods such as LASSO have been proposed to handle high-dimensional covariates, 
they generally rely on exact or approximate sparsity, under which the relevant signal is concentrated in a relatively small subset of predictors.
Recent work by \cite{adamek2022lasso,adamek2023_hdlp} proposed that researchers use a debiased or desparsified LASSO in time series and LP with a weaker form of sparsity. Their results are formulated in terms of an effective sparsity level.
For example, \cite{giannone2021illusion} argue that predictive information in macroeconomic datasets is often distributed across many variables rather than concentrated in a small subset. This observation underscores the need for approaches that remain effective when 
predictive content is spread across many covariates rather than concentrated in a few.

Accordingly, this paper develops a high-dimensional LP approach based on
the orthogonal greedy algorithm with a high-dimensional Akaike information criterion (OGA+HDAIC), originally proposed by \cite{ing2020}. 
The proposed framework does not require exact sparsity: predictive content may be distributed across many covariates, with many nonzero coefficients that gradually decay in magnitude. This coefficient-decay condition is a form of approximate sparsity, but it differs in formulation from the approximate-sparsity conditions commonly used in LASSO-based analyses, which are typically expressed through the existence of a sufficiently accurate $s$-sparse approximation and yield rates indexed by the effective sparsity level $s$ \citep{belloniInferenceTreatmentEffects2014}. 
Here, the coefficient class is characterized by an algebraic decay exponent, and OGA+HDAIC selects the model size along the greedy path. The rate results are indexed by this exponent and established for a path cap calibrated to it, rather than being stated in terms of an effective sparsity level.
Such structures arise naturally in time series settings, for example when predictive contributions are spread across many lags, making the approach particularly relevant for high-dimensional economic time series applications.

In addition to accommodating 
predictive content spread across many covariates,
the proposed procedure offers a practical advantage relative to factor-based approaches such as principal component (PC) methods \citep{stock2002forecasting}. It works directly with observed covariates rather than latent summaries, which makes the selected controls straightforward to inspect. This transparency is a secondary feature; the primary objective remains valid estimation and inference for impulse responses. 
From a theoretical standpoint, this method differs from the earlier cross-sectional framework elaborated with my coauthors in \citet{cha2023inference}, where we develop a double machine-learning approach under independent sampling. The present analysis focuses instead on time series regression for impulse response estimation, where temporal dependence and lagged regressors play a central role. Building on the OGA+HDAIC framework of \cite{ing2020}, the method is extended to dependent data through the functional dependence framework of \cite{wuNonlinearSystemTheory2005} and \cite{zhang2017gaussian}.
OGA+HDAIC is applied once to the shock equation and separately to each horizon-specific outcome equation, and the LP at each horizon is estimated using the union of the selected controls.
This double-selection structure \citep{belloniInferenceTreatmentEffects2014} delivers valid inference while allowing the two nuisance equations to select different controls.
In many empirical applications of LP, identification relies on conditioning assumptions or on sufficiently rich lag structures. In such settings, omitted variable concerns may arise either from incomplete control sets or from the use of finite-lag approximations to address invertibility. The ability to accommodate high-dimensional controls allows researchers to flexibly incorporate richer observable information when assessing these assumptions. Importantly, the proposed method does not alter the underlying identification strategy but provides a framework for estimation and inference when large control sets are employed.
Section \ref{sec:overview} introduces the general framework and illustrates how the proposed method can be implemented in the LP setting for impulse response estimation.

To assess the finite-sample properties of the proposed estimator, I conduct simulation experiments that vary the location and structure of the high-dimensional nuisance component. The simulations compare the proposed method with conventional low-dimensional LP methods, factor-based approaches, and regularization-based estimators. The results show that high-dimensional controls can improve robustness when they target contamination in the observed shock and improve precision when they absorb predictable outcome-equation variation, while factor-based methods are naturally advantaged when the nuisance component is explicitly factor-driven.

In the empirical application, I apply the proposed method to 
a monetary policy shock design in a high-dimensional setting based on the FRED-MD database \citep{mccracken2016fred}. This exercise illustrates the feasibility of the proposed approach in a $p > T$ environment and compares its performance with that of other high-dimensional impulse response methods.

\textbf{Related Literature. }
The literature on LP provides important insights into estimating dynamic causal effects. Following the initial introduction by \citet{jorda2005estimation} as a tool for estimating impulse responses, \cite{angrist2018_lpstringtheory} developed an identification scheme for dynamic causal effects, further developed into a difference-in-differences (DiD) approach in \cite{dube2023local}. \cite{MontielMikkel2021_ecta_lpinference} proposed a lag augmentation, assuming a finite lag order for the DGP, and \cite{xu2023local} put forth a uniform inference framework for LP with increasing lag order.
Most recently, \cite{adamek2023_hdlp} proposed adjustments for high-dimensional settings.


A related literature considers high-dimensional time series regression under dependence more broadly (\citealp{babiiEstimationHACbasedInference2019}; \citealp{chernozhukov2021_lassotimespace}; \citealp{babii2022machine}; see also the review by \citealp{masini2023machine}). 
Some approaches relax the assumption of exact sparsity: Elastic net and ridge regression \citep{carrascoInsampleInferenceForecasting2016} accommodate dense signals through $\ell_2$ penalization, and \citet{chernozhukovLavaAttackRecovery2017} explicitly decompose the signal into its sparse and dense components (see also \citealp{fanBridgingFactorSparse2023} and the review by \citealp{beyhumTestingSparseIdiosyncratic2024}).

Beyond LASSO-based approaches, the literature has also developed high-dimensional methods that do not rely on sparsity assumptions. A prominent and well-established line of work is the factor-based approach, including PC methods (\citealp{stock2002forecasting}) and factor-augmented VARs. These methods, which exploit pervasive comovement across macroeconomic series to extract a small number of common factors that summarize the information contained in large datasets, have become central tools in empirical macroeconomics.
The present paper instead targets settings where predictive content may be distributed across many observed covariates without requiring that this information be summarized by a dominant factor structure.




The estimation procedure in this paper builds on \cite{ing2020}, who proposes the OGA+HDAIC method for high-dimensional regression under the different notion of coefficient decay that allows signals to be distributed across many covariates rather than concentrated in a small subset. Around the same time, \cite{kozbur2020analysis} developed a testing-based forward selection approach for high-dimensional regression, and \cite{han2020high} contributed with a LASSO-based approach for time series regression, although both of these methods retained the sparsity assumptions.
Relatedly, \cite{cha2023inference} develop inference for a low-dimensional parameter of interest in a cross-sectional partially linear framework under independent sampling. The present paper extends this line of work to dependent time series settings and to the LP environment.

Separately, recent work has proposed practical tools tailored to the LP setting. \cite{barnichon2019impulse} utilize splines to smooth impulse responses. \cite{dinh2024random} apply a random subspace method to LP, with the aim---similar to this paper's---of accommodating rich information sets beyond sparsity.

Against this background, the present paper combines OGA+HDAIC with double selection and functional dependence tools to obtain valid LP
inference under coefficient-decay conditions.

The remainder of this paper is organized as follows. Section \ref{sec:overview} presents the notation and implementation details of the proposed method.
Section \ref{sec:theory} outlines the theoretical framework and main results. 
Section \ref{sec:simulation} presents its finite-sample performance through simulation studies. 
Section \ref{sec:applications} 
offers an empirical application on the transmission of monetary policy shocks in a high-dimensional setting. Finally, Section \ref{sec:conclusion} concludes.

%% file: sections/2_Overview.tex
I first introduce a general estimation framework and then provide further details on the estimation procedures later in this section. Consider the following $h$-step-ahead LP regression model:
\begin{align}\label{eq:lpequation_slowfast}
    y_{t+h} = \beta_h x_t + \boldsymbol{r}_{t}'\delta_{h,0} + \sum_{\ell=1}^L \boldsymbol{z}_{t-\ell}' \delta_{h,\ell} + u_{t,h},
\end{align}
where $y_{t+h} \in \mathbb{R}$ is the $h$-step-ahead response variable, $x_t \in \mathbb{R}$ is the innovation, $\boldsymbol{r}_{t} \in \mathbb{R}^{d_r}$ is contemporaneous controls, and $\boldsymbol{z}_{t} \in \mathbb{R}^{d_z}$ is lagged controls. 
$t=1,\dots,\bar{T}$ indexes the observations, $h=1,\dots,H_{\max}$ indexes horizons, and $\ell=1,\dots,L$ indexes lags. 
This is a general representation of LP from the literature, as introduced in \cite{plagborg2021local}.
In the spirit of \cite{jorda2005estimation}, I fix $h$ and focus on each horizon of interest.

We are interested in the response of $y_{t+h}$ with respect to a shock in $x_t$, and our parameter of interest $\beta_h$ is defined as
the coefficient on $x_t$ in the linear projection of $y_{t+h}$ onto $(x_t, \boldsymbol{r}_t, \boldsymbol{z}_{t-1},\dots,\boldsymbol{z}_{t-L})$. This corresponds to the standard LP estimand.

Below, I illustrate the framework through two applications. The first revisits impulse response analysis using LP, serving as the lead example throughout the paper. The second highlights the applicability of the method to high-dimensional partial linear regression models in time series settings.

\subsection{Applications}\label{sec:overview:irf}
\subsubsection*{Impulse response analysis with local projections}
Suppose, following \citet{plagborg2021local}, that a scalar shock \(x_t\) has been identified through a recursive ordering, a narrative design, or another procedure that delivers an observed shock series. The horizon-\(h\) impulse response of \(y_{t+h}\) to this shock, denoted by \(\beta_h\), is then recovered as the coefficient on \(x_t\) in the corresponding horizon-specific projection, conditional on an appropriate set of controls. The identified shock is common across horizons; only the outcome projection varies with \(h\).

The proposed framework is therefore aimed at identification strategies that deliver a scalar shock for use in LP. System-based procedures that impose joint restrictions across multiple shocks or horizons are not considered here. Even after the shock has been identified, however, estimation may require a rich set of contemporaneous and lagged controls---either to approximate the relevant conditioning information or remove predictable components of the observed shock. The econometric problem this paper addresses is how to achieve valid estimation and inference in the presence of such high-dimensional nuisance covariates.

A qualification concerns proxy-shock settings. If \(x_t\) is only a noisy proxy for the underlying structural shock, the present theory applies only when the included controls and maintained assumptions are sufficient to remove the relevant contamination. General proxy or noninvertible settings requiring instrumental variable methods are outside the scope of the paper.

\subsubsection*{High-dimensional partially linear regression models}
More generally, the local projection equation is a special case of the high-dimensional partially linear time series regression
\[
    y_t = \beta x_t + g(\boldsymbol w_t) + u_t,
    \qquad
    g(\boldsymbol w_t)=\boldsymbol w_t'\gamma,
\]
where \(x_t\in\mathbb R\) is the regressor of interest and \(\boldsymbol w_t\in\mathbb R^p\) is a potentially high-dimensional dictionary of nuisance covariates. From this perspective, \(\beta_h\) is a low-dimensional target parameter, and the proposed procedure provides inference in the presence of a high-dimensional linear nuisance component. 
Similar structures arise in other time series regression settings. For example, in predictive regressions, one may study the coefficient on a particular predictor of interest, such as inflation or financial conditions, for future outcomes while controlling for a large set of additional covariates. Likewise, in Granger causality-type exercises the parameter of interest is the coefficient on a lag of a candidate predictor, with many other lags and variables entering as nuisance controls.

\subsection{Estimation Procedures}\label{sec:overview:est}
For notational simplicity, stack the covariates except for $x_t$ into $\boldsymbol{w}_t = (\boldsymbol{r}_{t},\boldsymbol{z}_{t-1},\dots,\boldsymbol{z}_{t-L})\in \mathbb{R}^{p}$, where $p=d_r + L \cdot d_z$, and write
\begin{align}\label{eq:lp_baseline}
    y_{t+h} = \beta_h x_t + \boldsymbol{w}_t'\beta_{-h} + u_{t,h},
\end{align}
where $\beta_{-h} = (\delta_{h,0}',\delta_{h,1}',\dots,\delta_{h,L}')'$.

A straightforward approach to \eqref{eq:lp_baseline} with high dimensionality would be to apply a regularization method such as LASSO.
However, regularization generally induces bias in the coefficient of interest. A large literature proposes debiasing strategies to address this issue, including node-wise regression and related orthogonalization methods (e.g., \citealt{van2014asymptotically}; \citealt{zhang2014confidence}).

An alternative approach is double selection, whereby a similar goal is achieved through selecting covariates relevant for either the outcome or the regressor of interest, which ensures the orthogonality condition is met \citep{belloniInferenceTreatmentEffects2014}. In this paper, I adopt a double selection procedure combined with the OGA+HDAIC selection rule proposed by \cite{ing2020}.
OGA constructs a greedy ordering of the covariates, beginning with the variable having the greatest univariate explanatory power and subsequently selecting variables according to their correlations with the current residual.
The final model size is then chosen by minimizing the high-dimensional AIC over the OGA path.
The selections come from the following regression models:
\begin{align}
y_{t+h} = \boldsymbol w_t' \lambda + e_{t,h},\label{eq:yonw}\\ 
x_t = \boldsymbol w_{t}' \pi + v_{t}. \label{eq:xonw}
\end{align}
The concept involves applying OGA+HDAIC on both equations \eqref{eq:yonw} and \eqref{eq:xonw}. To remove the regularization bias, I use the union as the selected covariates to finally run the LP regression in \eqref{eq:lp_baseline}. 
Because the auxiliary regression used to recover $x_t$ does not depend on the projection horizon, the OGA+HDAIC selection for equation \eqref{eq:xonw} is performed only once.
For a clear presentation of the OGA+HDAIC procedure, I fix some notation below.
\begin{not*}
Fix a horizon $h$. Let $t$ index observations and $j$ index covariates. 
Let $T := \bar{T}-h-L$ denote the effective sample size.
Let $\boldsymbol w_j := (w_{t,j})_{t=1}^T$ denote the time series of the $j-$th covariate and $\boldsymbol{w}_t := (w_{t,j})_{j=1}^p$ denote the vector of covariates at time $t$. Let $J \subset \{1,\ldots,p\}$ be a set of covariate indices, and write $\boldsymbol w_t(J) := (w_{t,j})_{j\in J}$. 
Stack observations to define $W(J) := (\boldsymbol w_t(J))_{t=1}^T$ and $\boldsymbol y := (y_{t+h})_{t=1}^T$. 
Let $P_J$ denote the projection onto the column space of the original covariates $\{w_j : j \in J\}$.
\end{not*}

The OGA+HDAIC procedure is applied separately to the outcome equation \eqref{eq:yonw} and the shock equation \eqref{eq:xonw}. For concreteness, consider the regression of \(y_{t+h}\) on the candidate controls \(w_t\). In OGA, an ordered list of covariates is constructed recursively. Starting from the empty set, the procedure selects the covariate with the largest scale-normalized correlation with the current residual. After each selection, I update the current residual by projecting \(y_{t+h}\) on the covariates selected so far, and the next covariate is chosen by the same residualized marginal-correlation criterion.

The first-step score is
\begin{align}\label{eq:def:mu_empty}
    \mu_j(\emptyset)
    =
    \frac{w_j'y}{\sqrt{T}\|w_j\|},    
\end{align}
and the residualized score at step \(m\) is
\[
    \mu_j(\widehat J_{m-1})
    =
    \frac{w_j'(I-P_{\widehat J_{m-1}})y}
    {\sqrt{T}\|w_j\|},
    \qquad
    j\notin \widehat J_{m-1}.
\]
The next covariate is selected as
\[
    \widehat j_m
    =
    \arg\max_{j\in[p]\setminus \widehat J_{m-1}}
    \left|\mu_j(\widehat J_{m-1})\right|,
    \qquad
    \widehat J_m=\widehat J_{m-1}\cup\{\widehat j_m\}.
\]
The resulting ordered path is then truncated by the HDAIC
, defined as
\begin{align}\label{eq:def:HDAIC}
    \text{HDAIC}(J) = T \log({\widehat{\sigma}}_J^2) + C^* |J| \kappa_T,
\end{align}
where $\widehat{\sigma}_J^2 = \boldsymbol y'(I-P_J)\boldsymbol y /T$. 
Here \(\kappa_T\) is a high-dimensional penalty sequence. The baseline implementation takes \(\kappa_T=\log p\), corresponding to the light-tail concentration rate formalized in Section~\ref{sec:theory}. The theoretical results also allow the more general calibration \(\kappa_T=T a_T^2\), where \(a_T\) is the maximal-inequality rate implied by the tail and dependence conditions.
The number of covariates to be included in the model is the one that minimizes the information criterion,
\begin{align} \label{eq:def:argminHDAIC}
    \widehat{m} = \arg\min_{1\le m\le M_T^*} \text{HDAIC}(\widehat{J}_m).
\end{align}
Denote the chosen set of covariates as $\widehat{J}^{[1]} := \widehat{J}_{\widehat{m}}$. Next, repeat the OGA+HDAIC for \eqref{eq:xonw} and obtain the set of covariates $\widehat{J}^{[2]}$.
Finally, denote the union set as $\widetilde{J} = \widehat{J}^{[1]}\cup \widehat{J}^{[2]}$ and run the LP regression \eqref{eq:lp_baseline} with the selected covariates:
\begin{align*}
    y_{t+h} = \beta_h x_t + \wt(\widetilde{J})'\beta_{-h} + u_{t,h}.
\end{align*}
The least squares estimator for $\beta_h$ in this final model is our proposed estimator.
Let \(\widehat v_t\) denote the residual from regressing \(x_t\) on \(\bw_t(\widetilde J)\), and let \(\widehat u_{t,h}\) denote the residual from the final local projection regression using the same selected union.
Define $\widehat{\tau}^2 = \frac{1}{T}\sum_{t=1}^T \widehat v_{t}^2 .$ The variance estimator is then defined as
\begin{align}
    \widehat{\sigma}^2 =& \frac{1}{\hat{\tau}^4}\widehat{\Omega},\label{eq:def:sigma}\\
\widehat{\Omega}_h
=&\ \frac{1}{T}\sum_{t=1}^T
\widehat v_t^{\,2}\widehat u_{t,h}^{\,2}
+
2\sum_{\ell=1}^{L_h}
K\!\left(\frac{\ell}{L_h}\right)
\frac{1}{T}\sum_{t=\ell+1}^{T}
\widehat v_t\widehat u_{t,h}
\widehat v_{t-\ell}\widehat u_{t-\ell,h}
\label{eq:def:Omega}
\end{align}
where $\widehat\Omega$ is a kernel-based heteroskedasticty- and autocorrelation-consistent (HAC) estimator with kernel function $K(\cdot)$
and bandwidth parameter $L_h$.
The consistency of this variance estimator is established in Theorem~\ref{thm:variance}.

The proposed algorithm is detailed in Algorithm \ref{alg:double} after a remark on tuning parameters.

\begin{rem}



The HDAIC step involves two implementation constants, \(M_T^*\) and \(C^*\); Online Appendix~S2 provides details. 
The path cap $M_T^*$ limits the number of OGA steps, and $C^*$ scales the HDAIC penalty. With the theoretical calibration $\kappa_T=Ta_T^2$ and $M_T^*=\min\{p,\lfloor c_Ma_T^{-1/\delta}\rfloor\}$, the rate results hold for each fixed $C^*>0$. They do not identify a unique finite-sample choice.
I use $C^*=2$ as the baseline implementation value and report sensitivity results in Online Appendix S4.1. Online Appendix S2 distinguishes the theoretical calibration from the practical path cap used in the simulations.

\end{rem}

\begin{algorithm}{[Double OGA+HDAIC]}\label{alg:double}
\begin{enumerate}
\item Apply the OGA+HDAIC criterion in \eqref{eq:yonw} to the outcome nuisance regression, obtaining \(\widehat J_{y}\).
\item Apply the same procedure to the shock equation in \eqref{eq:xonw}, obtaining \(\widehat J_x\).
\item Set \(\widetilde J=\widehat J_x\cup\widehat J_{y}\) and refit the shock and final LP equations on \(W(\widetilde J)\).
\item Construct the HAC variance estimator in \eqref{eq:def:sigma}--\eqref{eq:def:Omega}.
\end{enumerate}
\end{algorithm}


\begin{rem}
The variable selection step for the $x_t$ equation in Step 2 does not depend on the horizon $h$.
Because the shock $x_t$ is common across all horizons, the selected set $\hat J_x$ is identical for all $h$.
Horizon specificity arises only through the response equation for $y_{t+h}$.
\end{rem}

\begin{rem}

The OGA ordering step is coordinate dependent: In finite samples, the greedy selection path can change under general linear transformations of the candidate controls. Conditional on a fixed selected control span, the post-selection ordinary least squares (OLS)/Frisch-Waugh-Lovell (FWL) estimator is invariant to nonsingular reparameterizations of that span, but transformations of the candidate dictionary can lead to different selected models. This is inherent to coordinate-wise selection methods. In applications, continuous covariates should be standardized before selection, and robustness to substantively meaningful alternative codings should be assessed; see also \citet{kolesar2023fragility}.
\end{rem}

\begin{rem}
The proposed procedure does not employ cross-fitting. In time series settings, cross-fitting would require dependence-aware sample splitting, such as blocking, which can be inefficient in typical macroeconomic samples. The double-selection step is therefore used without cross-fitting.
\end{rem}


\subsection{Extension: Always-Included Controls}\label{sec:overview:always}
In many empirical applications, a subset of controls is prespecified and included in all regression specifications for identification or institutional reasons. 
Examples include prespecified contemporaneous variables and lagged controls that are required by the shock definition or empirical design.

Let \(d_t\) collect the fixed-dimensional controls that are always included, and let \(w_t\) collect the remaining candidate controls subject to selection. Thus \(d_t\) may contain prespecified contemporaneous controls and lags, and \(w_t\) contains the corresponding selectable controls. The LP equation is
\begin{align}
    y_{t+h}     =     \beta_h x_t    +    d_t'\theta_h    +    w_t'\gamma_h    +    u_{t,h}.
    \label{eq:dt_structural}
\end{align}
The auxiliary regressions are
\begin{align}
    y_{t+h}    &=    d_t'\theta_h^{(y)}    +    w_t'\lambda    +    e_{t,h},    \label{eq:dt_yonw} \\
    x_t    &=    d_t'\theta^{(x)}    +    w_t'\pi    +    v_t .
    \label{eq:dt_xonw}
\end{align}
The procedure is the same as Algorithm~\ref{alg:double}, except that \(d_t\) is included unpenalized in all regressions and selection is applied only to \(w_t\), equivalently after partialling out \(d_t\). Corollary~\ref{cor:single} formalizes this residualization argument.

\begin{rem}[Choice of always-included controls]
The always-included set should contain variables that are required by the researcher's shock definition or identifying strategy. These variables should not be treated as candidate controls subject to selection because omitting them may change the estimand or invalidate the intended identifying restriction. Instead, they are included unpenalized in the shock equation, the horizon-specific outcome equation, and the final LP regression.

The theoretical result for always-included controls is stated for a fixed-dimensional vector \(d_t\) satisfying the same type of stationarity, moment, dependence, and nonsingularity conditions imposed on the remaining variables. Thus, the result covers fixed-dimensional macroeconomic controls and their lags, provided they are stationary after the transformations used in the simulation or application. Deterministic trends, unit-root variables, or other nonstationary components require separate assumptions and are not covered by Corollary~\ref{cor:single} as stated.
\end{rem}

%% file: sections/4_Theory.tex
This section presents the assumptions and main theoretical results.
I first introduce some notation used throughout the section.
\begin{not*}
    For a scalar random variable $x$, let $\|x\|_q := (E|x|^q)^{1/q}$ denote the $L_q$ norm.
    For a vector $\boldsymbol x=(x_1,\ldots,x_p)'$, let $\|\boldsymbol x\|_q := (\sum_{j=1}^p |x_j|^q)^{1/q}$ and $\|\boldsymbol x\|_\infty := \max_j |x_j|$.
\end{not*}
\subsection{Assumptions}
The definitions of the variables and parameters come from the baseline models \eqref{eq:lp_baseline} -- \eqref{eq:xonw}.
For notational simplicity, I suppress the horizon index $h$ throughout and focus on inference at a fixed horizon.
\begin{asm}[Functional dependence] \label{asm:FD}
Assume the following for a given horizon $h$.
Let $\{\zeta_t\}_{t\in\mathbb Z}$ be an i.i.d.\ innovation sequence.
\begin{enumerate}[label=(\alph*)]
\item 
The scalar processes $\{x_t,\, w_{t,1},\ldots,w_{t,p},\, u_{t,h},\, e_{t,h}, v_t\}$ are stationary and admit causal representations as measurable functions of a common innovation sequence,
$z_t = g_z(\mathcal F_t)$ with $\mathcal F_t := (\ldots,\zeta_{t-1},\zeta_t),$
where $z_t$ denotes a generic scalar element of $\{x_t,\, w_{t,1},\ldots,w_{t,p},\, u_{t,h},\, e_{t,h}, v_t\}$.\label{asm:FD:FD}

For $i\ge0$, let $z_t^{(i)}$ denote the coupled version of $z_t$ obtained by replacing
$\zeta_{t-i}$ with an independent copy $\zeta_{t-i}'$.
Define the functional dependence coefficient:
\begin{equation}
\delta_{i,q}(z)
\;:=\;
\| z_t - z_t^{(i)} \|_q,
\label{eq:fd_coeff}
\end{equation}
and the dependence-adjusted norm \citep{zhang2017gaussian}:
\begin{equation}
\|z\|_{q,\alpha}
\;:=\;
\sup_{m\ge0} (m+1)^{\alpha}
\sum_{i=m}^{\infty} \delta_{i,q}(z).
\label{eq:dan}
\end{equation}
Let $\epsilon_t$ denote a generic error term appearing in equations \eqref{eq:lp_baseline}--\eqref{eq:xonw}. 
Assume that 
\item there exists a constant $q>8$ such that 
$\sup_{1\le j \le p}\norm{w_{t,j}}_q \vee \norm{x_t}_q \vee \norm{\epsilon_t}_q <\infty$\label{asm:FD:moment}
\item there exist a constant $\alpha>1/2 - 1/q$ such that
$\sup_{1\le j\le p}\norm{w_j}_{q,\alpha} \vee \norm{x}_{q,\alpha} \vee \norm{\epsilon}_{q,\alpha} < \infty$. \label{asm:FD:dan}
\item $E[x_t u_{t,h}] =0$, $E[{w}_{t,j} \epsilon_t] =0$ for all $j=1,\dots,p$, $E[v_t u_{t,h}] = 0$, $E[v_t^2]>0$, and $E[e^2_{t,h}]>0$. 
In the always-included-controls specification, the residuals are defined after the partialling out of the fixed-dimensional controls \(d_t\) and additionally satisfy $E[d_tu_{t,h}]=0$, $E[d_tv_t]=0$.
\label{asm:FD:meanzero}
\end{enumerate}
\end{asm}

Assumption \ref{asm:FD} formalizes weak temporal dependence using the functional dependence framework of \cite{zhang2017gaussian}.
Part \ref{asm:FD:FD} imposes a causal representation driven by a common innovation sequence, which allows for general nonlinear and non-Gaussian dynamics.
Part \ref{asm:FD:moment} imposes sufficiently high finite moments for the Gram and score product processes used in the high-dimensional concentration argument; in particular, $q>8$ ensures that these products have more than four finite moments. This condition rules out arbitrarily heavy-tailed processes with fewer than eight finite primitive moments, but it is weaker than imposing sub-Gaussian or subexponential tails.
Part \ref{asm:FD:dan} controls the strength of temporal dependence through the dependence-adjusted norm, placing the process in the regime where Gaussian-type maximal inequalities apply.
Finally, part \ref{asm:FD:meanzero} imposes standard orthogonality restrictions used for the identification of structural shocks and construction of valid score equations, together with a nondegeneracy condition on the nuisance disturbances.

Assumption \ref{asm:FD} is a flexible notion of weak temporal dependence that allows for general nonlinear and non-Gaussian time series.
It requires that the effect of a shock $\zeta_{t-i}$ on the present value of a process decays at a polynomial rate as $i$ increases, with the parameter $\alpha$ governing the speed of this decay.

In contrast to classical mixing conditions, functional dependence does not impose restrictions on joint distributions across distant time indices.
Instead, dependence is quantified directly through the sensitivity of the process to past shocks, which makes the condition convenient to verify in a wide range of DGPs.
At the same time, Assumption \ref{asm:FD} is strong enough to deliver both high-dimensional concentration inequalities and standard central limit theorems.


As an illustration, consider a stationary Gaussian VAR$(L)$ process $Y_t = A_1 Y_{t-1} + \dots + A_L Y_{t-L} + \epsilon_t$ that admits a causal linear representation $Y_t= \sum_{\ell=0}^{\infty} \Phi_\ell \epsilon_{t-\ell}$ with geometrically decaying coefficients. Such processes satisfy Assumption \ref{asm:FD}. 
More generally, any linear process $z_t = \sum_{i=0}^\infty c_i \epsilon_{t-i}$ with i.i.d. innovations satisfies Assumption \ref{asm:FD} provided $\sum_{i=0}^\infty (i+1)^{\alpha} |c_i| < \infty$, a condition that accommodates MA$(\infty)$ representations.

The following is the assumption on the degree of sparseness of the underlying coefficients, $\lambda$ and $\pi$. 
\begin{asm}\label{asm:poly}
    Fix an arbitrary horizon $h \in \{1,\dots,H_{\max}\}$.
    Let $\xi$ denote a coefficient vector associated with the selectable covariates $w_t$ in the population regression at horizon $h$. 
    Assume that $\sum_{j=1}^p \xi_j^2 \le C$ and $|\beta_h| \le C$.
    Also assume that there exist $\delta > 1$ and $0<C_{\delta}<\infty$ such that, for any $J\subseteq \{1,\dots,p\}$,
        \begin{align}\label{eq:asm:poly}
            \sum_{j\in J}\abs{\xi_j} \le C_{\delta} \paren{\sum_{j\in J} \xi_j^2}^{(\delta-1)/(2\delta-1)},
        \end{align}
    where $C$ refers to a generic constant $0<C<\infty$. 
\end{asm}
Assumption~\ref{asm:poly} differs in formulation from the sparsity conditions commonly used in LASSO-based analyses. It does not restrict the number of nonzero coefficients or parameterize the coefficient class through a global \(\ell_1\)- or \(\ell_q\)-norm bound. Instead, it directly controls the concentration of coefficient magnitudes by imposing uniform bounds on their aggregate behavior over arbitrary index sets.

While the parameters need not have a natural order, we can express these conditions more simply by rearranging the parameters in descending order of their magnitude $|\xi_j|$. 
Denote the rearrangement as $|\xi_{(1)}|\ge |\xi_{(2)}|\ge \dots\ge |\xi_{(p)}|$. Then, one example of Assumption \ref{asm:poly} under this rearrangement is
\begin{align}\label{eq:polydecay}
    C_1 j^{-\delta} \le \abs{\xi_{(j)}} \le C_2 j^{-\delta},\: 0<C_1\le C_2<\infty, \:\delta>1,
\end{align}
where, from \eqref{eq:polydecay}, we call Assumption \ref{asm:poly} a polynomial decay case: If \eqref{eq:polydecay} holds for some $\delta > 1$, then \eqref{eq:asm:poly} holds for the same $\delta$, as shown in Lemma A.2 in \cite{ing2020}.
The exponent \(\delta\) is a fixed but unknown feature of the coefficient sequence. The theory allows any fixed \(\delta>1\), with larger values corresponding to faster coefficient decay under \eqref{eq:polydecay}.

Apart from \eqref{eq:polydecay}, another example of Assumption \ref{asm:poly} is
\begin{align}\label{eq:polysum}
    \sum_{j=1}^p \abs{\xi_j}^{1/\delta} < C ,\:0<C<\infty,\: \delta> 1,
\end{align}
which is a frequently adopted assumption in the high-dimensional literature, as in \cite{wang2014adaptive}.
Assume that \eqref{eq:polysum} holds for some $\delta$. Applying Hölder's inequality and
by setting $C_\delta=C^{\delta/(2\delta-1)}$, \eqref{eq:asm:poly} holds.
Thus, Assumption~\ref{asm:poly} includes coefficient vectors with uniformly bounded support as a special case and is also implied by standard weak sparsity conditions such as \eqref{eq:polysum}.


\begin{rem}
There are several notions of sparsity used in the high-dimensional literature.
Strong sparsity refers to exact sparsity, where only a small number of coefficients are nonzero.
Soft sparsity (or $\ell_q$ sparsity) relaxes this requirement by assuming bounded $\ell_q$ norms for some $0<q\le1$.
Another related concept is approximate sparsity, which requires that the coefficient vector be well approximated by a sparse vector, typically formalized through bounds on best-$s$ approximation errors or sufficiently fast decay of sorted coefficients (see, e.g., \citealp{chernozhukov2021_lassotimespace}, Comment 3.1).

As shown above, Assumption~\ref{asm:poly} is implied by standard \(\ell_q\)-type weak sparsity conditions, but it is formulated directly as a uniform restriction on coefficient mass over arbitrary index sets. It therefore does not require the coefficient vector to be parameterized through a prespecified support size or a best-\(s\) sparse approximation.

For a given decay exponent $\delta$, the relevant approximation scale is of order $a_T^{-1/\delta}$. The HDAIC criterion selects the model size along the OGA path without using $\delta$ directly, but the theoretical path cap is calibrated to this exponent. The results therefore establish rates under that calibration; they do not establish adaptation over unknown $\delta$ with a common decay-independent path cap.

\end{rem}

\begin{asm}\label{asm:add}
    Assume the following.
    \begin{enumerate}[label=(\alph*)]
        \item For some positive constant $C_1$ and $C_2$, it holds that\label{asm:add:IngA5}
        \begin{align*}
            \max_{1\le |J|\le C_1 a_T^{-1} ,i\notin J} \norm{\Gamma^{-1}(J)E[w_{t,i}\bw_t(J)]}_1 < C_2,
        \end{align*}
        where $\Gamma(J) = E[\bw_t(J)\bw_t(J)']$.
        \item Let $\Sigma = E[\sumt\wt\wt'/T]$. Assume $\max_{j\in[p]}\Sigma_{j,j}<C$ and $1/C \le \Lambda_{\min}$, where $\Lambda_{\min}$ is the minimum eigenvalue of $\Sigma$ and $0<C<\infty$. 
        In the always-included-controls setting of Section~\ref{sec:overview:always}, this condition is imposed on the augmented matrix $\Sigma_{aug} = E[\sumt \tilde d_t \tilde d_t'/T]$, where $\tilde d_t = (d_t',\wt')'$, in place of $\Sigma$.
        \label{asm:add:min_eig}
        \item Let \(q>8\) be the moment exponent in Assumption~\ref{asm:FD}. Define
        \[
            a_T
            :=
            \left(\frac{\log p}{T}\right)^{1/2}
            \vee
            p^{4/q}T^{2/q-1}.
        \]
        Assume $a_T\to0$ and $\sqrt T\,a_T^{2-1/\delta}=o(1)$, where \(\delta>1\) is defined in Assumption~\ref{asm:poly}. 
        In addition, for the consistency result for the HAC variance, assume $p^{4/q}T^{4/q-1}=O(1)$. \label{asm:add:rate}
    \end{enumerate}
\end{asm}

Assumption \ref{asm:add} imposes additional conditions on the covariates and on the growth rate of the dimension $dim(\wt)=p$. Part \ref{asm:add:IngA5} limits strong correlations between covariates, ensuring that the OLS coefficients of one covariate on any set of other covariates remain finite. Part \ref{asm:add:min_eig} is on the minimum eigenvalue of the Gram matrix, with a remark that this condition is on the population level, not on the sample matrix as in the restricted eigenvalue conditions in \cite{bickel2009simultaneous}. 

Part~\ref{asm:add:rate} defines the rate \(a_T\) at which the empirical Gram and score processes concentrate uniformly and restricts the joint growth of \(p\) and \(T\) through it. The rate has two components: the familiar Gaussian-type term \((\log p/T)^{1/2}\) and a polynomial-tail term \(p^{4/q}T^{2/q-1}\) reflecting that the Gram and score products possess only \(q/2>4\) finite moments under Assumption~\ref{asm:FD}\ref{asm:FD:moment}. Lemma~\ref{lem:gram-score} verifies that the relevant empirical processes are uniformly \(O_p(a_T)\) under this definition. The requirement \(\sqrt T\,a_T^{2-1/\delta}=o(1)\) ensures that the first-stage prediction error is small enough for the nuisance estimation remainder to be negligible in the \(\sqrt T\)-normal approximation.
\begin{rem}
The two components of \(a_T\) arise from Nagaev-type maximal inequalities for functionally dependent processes \citep[e.g.,][]{WuWu2016EJS,zhang2017gaussian}, in the spirit of the finite-moment concentration arguments of \citet{babii2022machine}. For a family of \(N\) centered stationary processes with moment order \(r>2\) and bounded dependence-adjusted norms, the maximum of the empirical means is \(O_p\{(\log N/T)^{1/2}\vee N^{1/r}T^{1/r-1}\}\). The Gram products \(w_{t,j}w_{t,k}\) form \(N=p^2\) processes of moment order \(r=q/2\) by Lemma~S1, giving the term \(p^{4/q}T^{2/q-1}\); the \(p\) score products contribute \(p^{2/q}T^{2/q-1}\), which is dominated. 

Two special cases are instructive. If \(p\) grows subexponentially and the polynomial term is dominated---as occurs, for instance, when \(p=O(T^{\kappa})\) with \(\kappa<(q-4)/8\)---then \(a_T=(\log p/T)^{1/2}\) and Part~\ref{asm:add:rate} reduces to the familiar condition \(\log p=o\bigl(T^{(\delta-1)/(2\delta-1)}\bigr)\). Under heavier tails (smaller \(q\)) or faster growth of \(p\), the polynomial term dominates and the admissible growth of \(p\) is correspondingly more restrictive. The moment condition \(q>8\) in Assumption~\ref{asm:FD} is what guarantees the product processes have more than four finite moments, so that the polynomial component vanishes for some polynomial growth of \(p\); the assumption rules out very heavy-tailed processes but does not impose sub-Gaussian or subexponential tails.
\end{rem}

\begin{rem}[Factor structures and the eigenvalue condition]
The minimum-eigenvalue condition in Assumption~\ref{asm:add}\ref{asm:add:min_eig} is compatible with approximate factor structures. 
Suppose $w_t=Bf_t+e_t$ and $ E[f_te_t']=0$.
Then $\Sigma_w=B\Sigma_fB'+\Sigma_e$, and, because \(B\Sigma_fB'\) is positive semidefinite, $\lambda_{\min}(\Sigma_w) \geq \lambda_{\min}(\Sigma_e).$
Thus, the condition is compatible with factor structures whenever the idiosyncratic covariance is bounded away from singularity. 
\end{rem}


\subsection{Inference}
In this subsection I establish the inference for the parameter of interest at a given horizon $h$.
For notational convenience, I fix horizon $h$ and suppress the dependence on $h$ whenever no confusion arises.
Throughout this subsection, let $\boldsymbol{y}$ and $W(J)$ be defined as in Section \ref{sec:overview:est}. 
In addition, let $\boldsymbol{x}$, $\boldsymbol{u}$, and $\boldsymbol{v}$ denote $T \times 1$ vectors collecting $\{x_t\}$, $\{u_{t,h}\}$, and $\{v_t\}$, respectively, and let $W$ denote the $T \times p$ matrix of selectable covariates, where $p = \dim(\boldsymbol{w}_t)$.
Then the models \eqref{eq:lp_baseline} -- \eqref{eq:xonw} can be written in a compact matrix form:
\begin{align}
    \boldsymbol{y} =& \beta_h \boldsymbol{x} + W \beta_{-h} + \boldsymbol{u},\label{eq:lp_vector_structural}\\
    \boldsymbol{y} =& W \lambda + \boldsymbol{e}, \label{eq:lp_vector_outcome}\\
    \boldsymbol{x} =& W \pi + \boldsymbol{v}. \label{eq:lp_vector_shock}
\end{align}
For the theoretical results below, OGA+HDAIC is implemented with $\kappa_T=Ta_T^2$ and $M_T^* = \min\left\{p,\left\lfloor c_Ma_T^{-1/\delta}\right\rfloor\right\}$, 
where $c_M>0$ and $C^*>0$ are fixed constants. Online Appendix~S2 provides further details on the calibration of these constants.

\begin{thm}[Error bounds for Double OGA+HDAIC] \label{thm:errbd}
    Fix an arbitrary horizon $h\in\{1,\ldots,H_{\max}\}$, and let \(\xi\in\{\lambda,\pi\}\) denote the coefficient vector in either the outcome or shock nuisance equation in \eqref{eq:lp_vector_outcome}--\eqref{eq:lp_vector_shock}. 
    Suppose Assumptions~\ref{asm:FD}--\ref{asm:add} hold.  
    Let \(\widehat\xi\) denote the least-squares refit obtained by applying OGA+HDAIC to the corresponding nuisance equation.
    Then
    \[
    \frac{1}{T}
    \left\|W(\widehat{\xi}-\xi)\right\|_2^2
    =
    O_p\left(a_T^{2-1/\delta}\right),
    \]
    where \(\delta>1\) is defined in Assumption~\ref{asm:poly} and $a_T$ is defined in Assumption~\ref{asm:add}\ref{asm:add:rate}.
\end{thm}
\begin{proof}
    The proof can be found in Appendix \ref{sec:Appendix:Thrm:errbd}.
\end{proof}

Theorem~\ref{thm:errbd} separates the approximation and stochastic components of the first-stage error. The exponent \(\delta\) measures how well the nuisance coefficients can be approximated by a low-dimensional subset, and \(a_T\) measures the cost of uniformly controlling the empirical Gram and score terms. In the baseline light-tail case, \(a_T=(\log p/T)^{1/2}\); more generally, \(a_T\) is determined by the maximal inequality available under the maintained tail and dependence conditions.

Although Theorem \ref{thm:errbd} itself is useful in deriving error bounds for the high-dimensional nuisance parameters, our parameter of interest is ${\beta}_h$ in equation \eqref{eq:lp_baseline}. The following theorem derives the asymptotic distribution of $\widehat{\beta}_h$ using the error bounds derived in Theorem \ref{thm:errbd}. The result is pointwise in the horizon $h$.

\begin{thm}\label{thm:single}
Suppose Assumptions~\ref{asm:FD}--\ref{asm:add} hold. Assume additionally that the long-run variance of the score process is nondegenerate: $0<\Omega_h := \sum_{\ell=-\infty}^{\infty} \mathrm{Cov}(v_0 u_{0,h}, v_\ell u_{\ell,h})<\infty$. Then
    \begin{align*}
        \sqrt{T}(\widehat{\beta}_h - \beta_h) \overset{d}{\to} N(0,\sigma_h^2),
    \end{align*}
    where $\sigma_h^2 := \Omega_h/\tau^4$ 
    and $\tau^2 = E[v_t^2]>0$.
\end{thm}
\begin{proof}
    The proof can be found in Appendix \ref{sec:Appendix:Thrm:single}.
\end{proof}

Applying the FWL theorem, I obtain the following corollary for the model with always-included controls in \eqref{eq:dt_structural}.
\begin{cor}\label{cor:single}
Suppose the model follows \eqref{eq:dt_structural} with a fixed-dimensional vector of prespecified controls \(d_t\). Let
\[
D=(d_1,\ldots,d_T)',
\qquad
M_D=I_T-D(D'D)^{-1}D',
\]
and define the sample residualized variables
\[
\widetilde{\boldsymbol y}_h=M_D\boldsymbol y_h,
\qquad
\widetilde{\boldsymbol x}=M_D\boldsymbol x,
\qquad
\widetilde W=M_DW.
\]
Let \(\widehat\beta_h\) denote the estimator obtained by applying Algorithm~\ref{alg:double} to the residualized variables \((\widetilde{\boldsymbol y}_h,\widetilde{\boldsymbol x},\widetilde W)\).

Further suppose that each scalar component of $d_t$ satisfies the moment and functional-dependence conditions in Assumption~\ref{asm:FD}\ref{asm:FD:FD}--\ref{asm:FD:dan}. Suppose also that Assumption~\ref{asm:add}\ref{asm:add:IngA5} holds for the population-residualized covariates
$\bar w_t=w_t-E[w_td_t']E[d_td_t']^{-1}d_t$.
Then, under Assumptions~\ref{asm:FD}--\ref{asm:add},
\[
\sqrt{T}(\widehat\beta_h-\beta_h)
\overset{d}{\to}
\mathcal N(0,\sigma_h^2),
\]
where \(\sigma_h^2\) is defined as in Theorem~\ref{thm:single} for the
corresponding residualized score.
\end{cor}
\begin{proof}
The proof can be found in Appendix~\ref{sec:Appendix:Thrm:cor}.
\end{proof}

Related HAC-type variance estimators for debiased high-dimensional procedures under dependence are studied in \cite{babiiHighdimensionalGrangerCausality2024} and \cite{adamek2022lasso,adamek2023_hdlp}. In the present setting, the object of interest is the long-run variance associated with the LP score process.
The following theorem establishes the validity of the proposed variance estimator.

\begin{thm} \label{thm:variance}
Fix \(h\), and let \(\widehat{\sigma}_h^2\) be defined in \eqref{eq:def:sigma}--\eqref{eq:def:Omega}. Suppose that \(K(\cdot)\) is bounded and symmetric, satisfies \(K(0)=1\), is continuous at zero, and 
absolutely integrable. Let \(L_h=L_{h,T}\) satisfy
\[
    L_h\to\infty,
    \qquad
    \frac{L_h^2}{T}\to0,
    \qquad
    L_h a_T^{\,1-1/\delta}\to0.
\]
Under Assumptions~\ref{asm:FD}--\ref{asm:add} and the nondegeneracy condition in Theorem~\ref{thm:single},
\[
    \widehat{\sigma}_h^2
    \overset{p}{\longrightarrow}
    \sigma_h^2,
\]
where \(\sigma_h^2\) is defined in Theorem~\ref{thm:single}.

\end{thm}
\begin{proof}
    The proof can be found in Appendix \ref{sec:Appendix:Thrm:variance}.
\end{proof}

The consistency result of Theorem~\ref{thm:variance} is stated for any kernel satisfying the regularity conditions in the theorem. These conditions are standard for HAC estimation and are satisfied by common choices such as the Bartlett and Parzen kernels. In the simulations and empirical application, I use the Parzen kernel, following \citet{babiiHighdimensionalGrangerCausality2024}. 
The Parzen kernel is bounded, symmetric, continuous at zero, compactly supported, and satisfies \(K(0)=1\).
Regarding the bandwidth, in the benchmark case \(a_T=(\log p/T)^{1/2}\), a polynomial bandwidth \(L_h\asymp T^c\) is admissible whenever $0<c<\frac{\delta-1}{2\delta},$ up to logarithmic factors. Thus the choice \(L_h\asymp T^{1/4}\) is covered only in the sufficiently sparse case \(\delta>2\); otherwise a slower-growing bandwidth is required.

With Theorems \ref{thm:single} and \ref{thm:variance}, I can now construct $100(1-\alpha)\%$ confidence intervals by means of a HAC estimator of the long-run variance defined in Theorem \ref{thm:variance}. Given a significance level $\alpha$, an asymptotic $100(1-\alpha)\%$ confidence interval $\widehat{\mathcal{I}}_{\alpha,h}$ is given by
\begin{align}
    \widehat{\mathcal{I}}_{\alpha,h} := [\widehat{\beta}_h - z_{\alpha}\hat{\sigma}_h/\sqrt{T}, \:\:\widehat{\beta}_h + z_{\alpha}\hat{\sigma}_h/\sqrt{T}],
\end{align}
where $z_{\alpha} = \Phi^{-1}(1-\alpha/2)$ is the $(1-\alpha/2)$ quantile of the standard normal distribution and $\hat{\sigma}_h$ is defined in Theorem \ref{thm:variance}.


%% file: sections/5_Simulation.tex
This section studies the finite-sample performance of the proposed double-selection LP procedure in a calibrated dynamic factor model (DFM) for quarterly US macroeconomic data. The stored calibration object was generated from the first-difference DFM workflow in the replication package \citep{montieloleaLpVarNbermaReplication2025} accompanying \citet{montieloleaLocalProjectionsVARs2025}. It is numerically identical to specification 6 in the replication package \citep{liLpVarSimulReplication2023} accompanying \citet{li2022local}. This calibration provides an empirically relevant environment with a large panel of macroeconomic controls and a known population impulse response.

The simulations address three questions. First, can high-dimensional controls improve robustness when the observed shock is contaminated by variables outside a conventional low-dimensional information set? Second, what is the role of outcome-equation controls when the observed shock is clean but the outcome contains predictable nuisance variation? Third, how does performance change when shock contamination is factor driven and therefore favorable to factor-augmented methods? The three designs correspond, respectively, to polynomial shock contamination, structured outcome nuisance variation, and factor driven shock contamination. Together, they clarify when coordinate selection is useful and when factor-based dimension reduction is better suited.

This section reports the main simulation results. Online Appendix~S3 provides the full replication details, including the exact variables observed by the econometrician, the construction of the structural shock, the restricted and full information sets, the contamination designs, Monte Carlo seeds, and estimator implementation details.

\subsection{Calibrated Dynamic Factor Model and Target Impulse Response}\label{sec:simulation-dfm}

Let \(Y_t\in\mathbb R^{p_y}\) denote the vector of observed macroeconomic variables, and let \(s_t\) denote the latent state vector. 
The calibration contains \(207\) observed series driven by six factors, with two lags in the factor dynamics and two lags in each idiosyncratic component. The matrices \(A,B,C,D\) below give the corresponding companion-form state-space representation.
\begin{align}
    s_t &= A s_{t-1}+B e_t, \label{eq:sim-state}\\
    Y_t &= C s_{t-1}+D e_t, \label{eq:sim-measurement}
\end{align}
where \(e_t\) collects the primitive innovations in the calibrated DFM. 
Online Appendix~S3 gives the corresponding primitive representation and replication details.

The response variable is core consumer price index (CPI) less food and energy, written as
\[
    y_t=a_y'Y_t,
\]
where \(a_y\) is the corresponding selection vector. The structural shock is a monetary policy shock defined as a fixed direction in the innovation vector,
\[
    \epsilon_t=q'e_t,
\]
where \(q\) is determined by the loading on the effective federal funds rate in the calibrated DFM and normalized to have unit length. The shock direction is fixed throughout the simulation study and is not reestimated or reselected across Monte Carlo replications.

Because the DFM is fully specified, the population impulse response of \(y_{t+h}\) to \(\epsilon_t\) is known. The target response is
\begin{align}
    \theta_0 &= a_y'Dq, \label{eq:sim-irf0}\\
    \theta_h &= a_y'CA^{h-1}Bq, \qquad h\geq1. \label{eq:sim-irfh}
\end{align}
All estimators are evaluated against the same population response \(\{\theta_h\}\).

The econometrician is not given the latent state \(s_t\), the primitive innovation vector \(e_t\), or the true shock direction \(q\). Instead, the econometrician observes the response \(y_t\), an observed shock measure \(x_t\), and a set of controls. The simulation designs below vary the relationship between \(x_t\), the structural shock \(\epsilon_t\), and the available controls. Thus, the target impulse response is fixed by the calibrated DFM; the econometric difficulty comes from the information set and the form of the nuisance component faced by the estimator.

\subsection{Information Sets and Estimation Methods} \label{sec:simulation-information}

The econometrician observes the response \(y_t\), an observed shock measure \(x_t\), and a set of macroeconomic controls. The latent state \(s_t\), the innovation vector \(e_t\), and the structural shock direction \(q\) are not observed by any estimator. I consider two main information sets.

The first is a low-dimensional information set, denoted by \(\mathcal I_R\), which is intended to mimic the information used in conventional small-scale LP. It contains five observed macroeconomic variables: core CPI less food and energy, the unemployment rate, the BAA--10-year Treasury spread, a house price index, and the federal funds rate. I refer to this set as the restricted information set. The second is the full information set, \(\mathcal I_F\), which contains the entire simulated macroeconomic panel \(Y_t\). Online Appendix~S3 reports the exact series mnemonics and column indices for both information sets.

The first group of estimators consists of conventional low-dimensional LP procedures using the restricted information set: standard LP \citep{jorda2005estimation}, lag-augmented LP \citep{MontielMikkel2021_ecta_lpinference}, and LP with AIC-selected lags and bias correction \citep{herbstBiasLocalProjections2024,montieloleaLocalProjectionsVARs2025}. The second group uses the full information set through aggregation or factor extraction, including factor-augmented LP (FALP), FAVAR \citep{bernankeMeasuringEffectsMonetary2005}, and random-subspace LP (RSLP, \citealp{dinh2024random}). The third group uses regularization or selection methods with the full information set, including debiased LASSO (dLASSO, \citealp{adamek2023_hdlp}), LAVA \citep{chernozhukovLavaAttackRecovery2017}, and the proposed double OGA procedure, denoted OGA in the simulation figures.

Unless otherwise stated, the estimators use four lagged control blocks. 
In particular, LALP follows the lag-augmentation specification of \citet{MontielMikkel2021_ecta_lpinference} with VAR order \(p=4\) and uses Eicker--Huber--White standard errors. AIC-BC instead selects a common lag order \(p\in\{1,\ldots,20\}\) for the restricted system and then applies the recursive finite-sample bias correction. FAVAR is estimated as a VAR(4). Each design uses \(T=200\) retained observations, projection horizons \(h=1,\ldots,12\), and \(1{,}000\) Monte Carlo replications.
Online Appendix~S3 gives remaining implementation details, tuning choices, standard-error construction, and bootstrap settings.

\begin{table}[htb!]
\begin{center}
\footnotesize
\setlength{\tabcolsep}{5pt}        
\renewcommand{\arraystretch}{1.10} %
\caption{Simulation estimators and information sets} \label{tab:simulation-methods}
\begin{tabular}{@{}p{0.18\textwidth}p{0.22\textwidth}p{0.55\textwidth}@{}}
\toprule
Method & Information set & Description \\
\midrule
OLS LP & Restricted & Conventional local projection using the observed shock and low-dimensional controls. \\
LALP & Restricted & Lag-augmented local projection using the restricted variables. \\
AIC-BC & Restricted & Local projections with AIC-selected lags and bias correction. \\
FALP & Full, via PCs & Factor-augmented local projection using six principal components and four lagged blocks of the response, shock, and factors. \\
FAVAR & Full, via PCs & Six-component factor-augmented vector autoregression of order four, with the observed shock ordered first. \\
RSLP & Full & Random-subspace local projection using randomly selected subsets of the full panel. \\
LAVA & Full & Sparse-plus-dense shrinkage estimator using the full panel. \\
dLASSO & Full & Debiased LASSO local projection using the full panel. \\
OGA & Full & Proposed double-selection procedure applied to the full panel in the shock and horizon-specific outcome equations. \\
\bottomrule
\end{tabular}
\end{center}
{\footnotesize \textit{Notes.} ``Restricted'' refers to the low-dimensional information set and ``Full'' to the full simulated macroeconomic panel. PC-based methods use principal components extracted from the full panel. LAVA is reported only for point estimation metrics and is omitted from the coverage and confidence-interval-width panels.
Online Appendix~S3 provides implementation details.}
\end{table}

\subsection{Design 1: Shock-Equation Contamination}\label{sec:simulation-design1}

The first design studies shock-equation contamination. The structural shock is
\(\epsilon_t=q'e_t\), but the econometrician observes
\[
    x_t=\epsilon_t+g_{t-1},
    \qquad
    g_t=\widetilde Y_t'\pi .
\]
The lag makes the contamination predetermined relative to the current structural innovation, so the target remains the DFM impulse response to \(\epsilon_t\). High-dimensional procedures can use the full panel to partial out \(g_{t-1}\), whereas low-dimensional procedures control only for the restricted information set.
In the baseline calibration, the entries of \(\pi\) follow the full-panel order and decay at rate \(j^{-1.5}\). Online Appendix~S3 reports the ordering and normalization.

\begin{figure}[htb!]
\begin{center}
    \caption{Shock contamination design}\label{fig:sim_shock}
    \includegraphics[width=\textwidth]{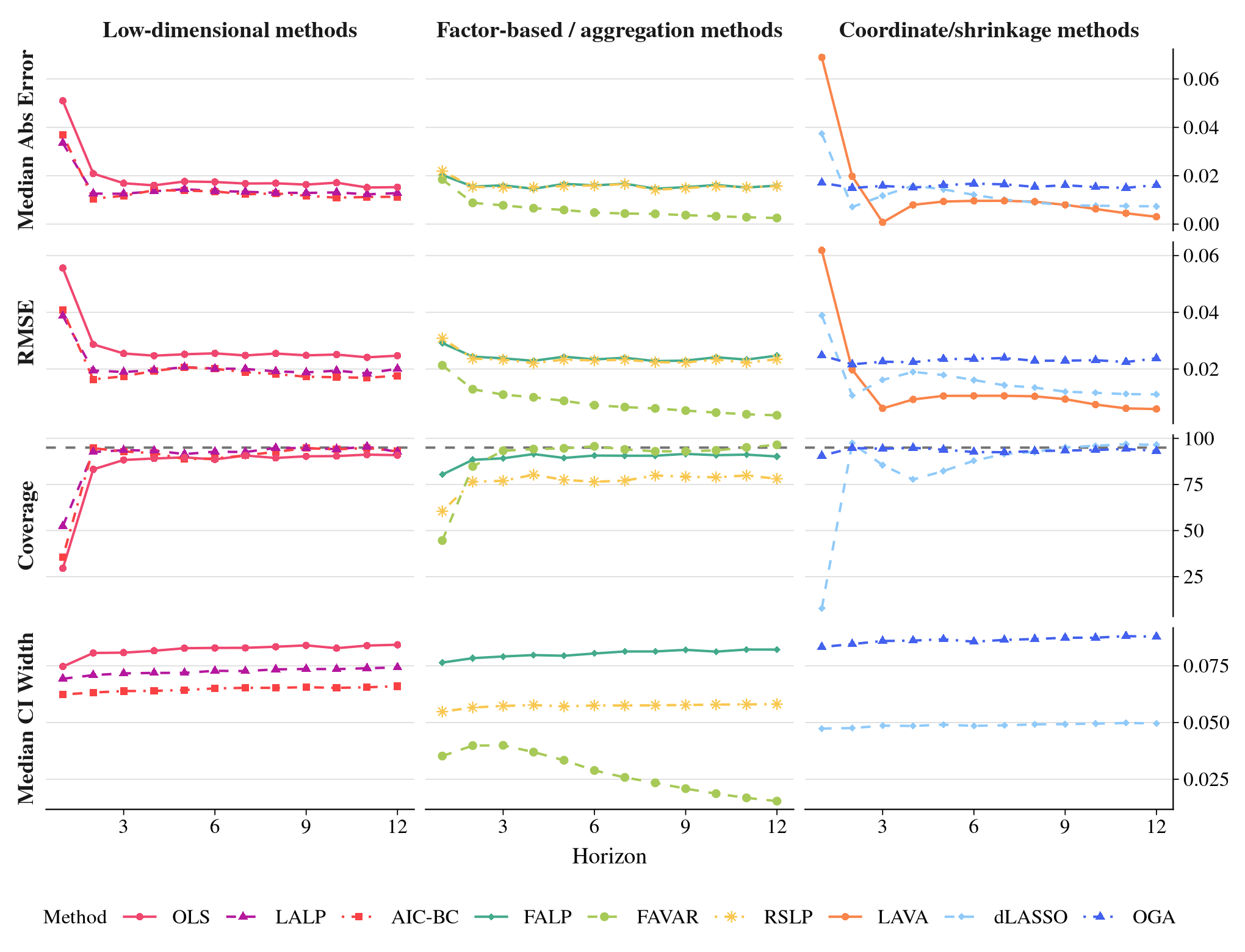}
\end{center}
    \vspace{-.4cm}
    {\footnotesize \textit{Notes.} The figure reports the median absolute error, RMSE, 95\% coverage, and median confidence-interval width across horizons. Method abbreviations are defined in Table~\ref{tab:simulation-methods}. LAVA is reported only in the point-error panels.}
\end{figure}
Figure~\ref{fig:sim_shock} reports the median absolute error and root mean-squared error (RMSE) for all methods and coverage and median confidence-interval width for the methods with reported interval estimates. In this shock-contamination design, 
the proposed procedure maintains coverage closest to the nominal level throughout most of the horizon range.
The low-dimensional LP estimators under-cover sharply at the first horizon, where the effect of omitted shock predictors is most pronounced. The factor-based and aggregation methods also exhibit weak first-horizon coverage, although their coverage generally improves at later horizons. Among the alternative variable-selection methods with interval estimates, dLASSO produces shorter intervals but under-covers at early horizons, whereas the proposed method trades wider intervals for more reliable coverage.

The results suggest that the improvement is not simply a consequence of giving the estimator access to the full high-dimensional panel. Other full-information methods also use the large panel but do not achieve the same first-horizon coverage. Moreover, Appendix Table~S4 shows that the proposed method selects only a small number of lagged covariates in this design, with a median of three selected in the shock equation and four in the final union. This suggests that the relevant mechanism is the shock-equation selection step: The procedure targets a small set of controls that help explain the contaminated shock measure, rather than mechanically adding many controls.

\subsection{Design 2: Outcome-Equation Nuisance}\label{sec:simulation-design2}

The second design studies the role of outcome-equation controls. Unlike in the first design, here, the observed shock measure is not contaminated: $ x_t=\epsilon_t.$
Instead, the outcome equation contains an additional structured high-dimensional nuisance component. 
For \(j=1,\ldots,K_z\), let
\[
    Z_{j,t}
    =
    \rho_z Z_{j,t-1}
    +
    \sqrt{1-\rho_z^2}\,\xi_{j,t},
    \qquad
    \xi_t\sim N(0,I_{K_z}),
    \qquad
    \xi_t\perp e_s \text{ for all } t,s ,
\]
where the parameter \(\rho_z\) controls the persistence of the nuisance variables.
The realized auxiliary series are projected on a constant and the realized structural shock and then standardized; see Online Appendix~S3.

Define the nuisance component
\[
    m_t = \frac{Z_t'\beta_z}{\operatorname{sd}(Z_t'\beta_z)},    \qquad
    \beta_{z,j}\propto j^{-\alpha_z},   \qquad  j=1,\ldots,K_z,
\]
and set
\[
    y_t^{obs} = y_t+\gamma_z m_{t}.
\]
Thus, the nuisance component enters the observed outcome contemporaneously. Because \(Z_t\) is persistent, however, it generates predictable variation in outcomes at subsequent projection horizons. The variables \(Z_t\) are included in the full information set but not in the restricted information set.

Because the auxiliary processes are generated independently of the DFM innovations, the added component has no population response to the structural shock \(\epsilon_t\), and the in-sample orthogonalization does not change the target.
Thus, the target impulse response remains the DFM response \(\{\theta_h\}\) defined in \eqref{eq:sim-irf0}--\eqref{eq:sim-irfh}. The design isolates the role of outcome-equation controls by adding predictable variation without changing the estimand. Methods using the full information set can exploit controls related to this component, while the restricted information set remains an informative low-dimensional benchmark. 
The main calibration sets $K_z=20, \rho_z=0.85, \alpha_z=1.5, \gamma_z=1$.

\begin{figure}[htb!]
\begin{center}
    \caption{Outcome nuisance design}\label{fig:sim_outcome}
    \includegraphics[width=\textwidth]{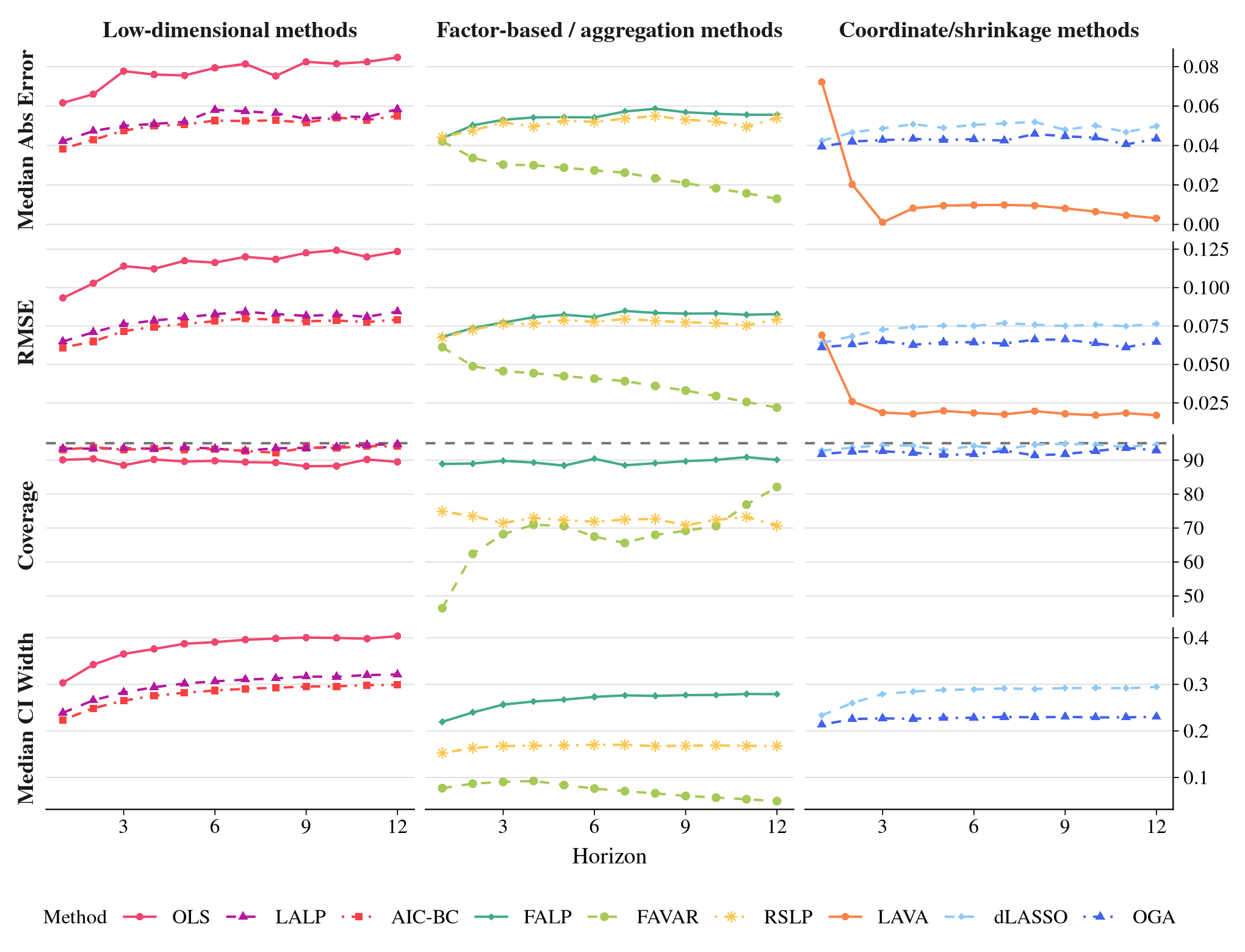}
\end{center}
        \vspace{-.4cm}
    {\footnotesize \textit{Notes.} The figure reports the median absolute error, RMSE, 95\% coverage, and median confidence-interval width across horizons. Method abbreviations are defined in Table~\ref{tab:simulation-methods}. LAVA is reported only in the point-error panels.}
\end{figure}


Figure~\ref{fig:sim_outcome} reports the results. Because the observed shock is clean, additional controls operate primarily by absorbing predictable variation in the outcome equation. Several methods improve substantially on the plain OLS specification in terms of point-estimation accuracy. In particular, FALP produces lower point-error measures and much narrower intervals than OLS while exhibiting broadly similar coverage, although both procedures remain below the nominal level.

The proposed method also produces low point-error measures and relatively short confidence intervals, but its coverage is a few percentage points below \(95\%\). AIC-BC and dLASSO generally attain better coverage, at the cost of wider intervals. The design therefore does not establish an unqualified efficiency gain for the proposed method. Instead, it illustrates the finite-sample trade-off between point accuracy, interval width, and coverage across alternative ways of exploiting predictable outcome variation. Appendix Table~S4 is consistent with the intended selection mechanism: The procedure selects a median of one control in the common shock equation and four controls in the horizon-specific outcome equations.

\subsection{Design 3: Factor-Driven Shock Contamination} \label{sec:simulation-design3}

The third design examines whether the structure of shock contamination matters. Unlike Designs~1 and~2, it makes the contamination explicitly factor-driven and is therefore favorable to factor-based procedures.

The observed shock measure is
\[
    x_t
    =
    \epsilon_t
    +
    \gamma_f \widetilde{b'F_{t-1}},
\]
where \(F_t\) denotes the vector of latent dynamic factors in the calibrated DFM, \(\widetilde{b'F_{t-1}}\) is the standardized version of \(b'F_{t-1}\), and \(b\) is a unit-length factor direction satisfying
\[
    \Lambda_R b=0,
    \qquad
    \|b\|=1.
\]

Here \(\Lambda_R\) denotes the factor-loading matrix for the variables in the restricted information set. The restriction \(\Lambda_Rb=0\) ensures that the contaminating direction has zero loading on the Restricted variables, while it remains represented in the full information set. The lag makes the contamination predetermined relative to the current structural innovation. The main calibration sets \(\gamma_f=2\). Online Appendix~S3 reports the construction of \(b\), the normalization of the contaminating component, and additional implementation details.

Because the nuisance component is low-rank and lies in the common-factor space, principal components extracted from the Full panel are naturally suited to approximate it.

\begin{figure}[htb!]
\begin{center}
    \caption{Factor contamination design}\label{fig:sim_factor}
    \includegraphics[width=\textwidth]{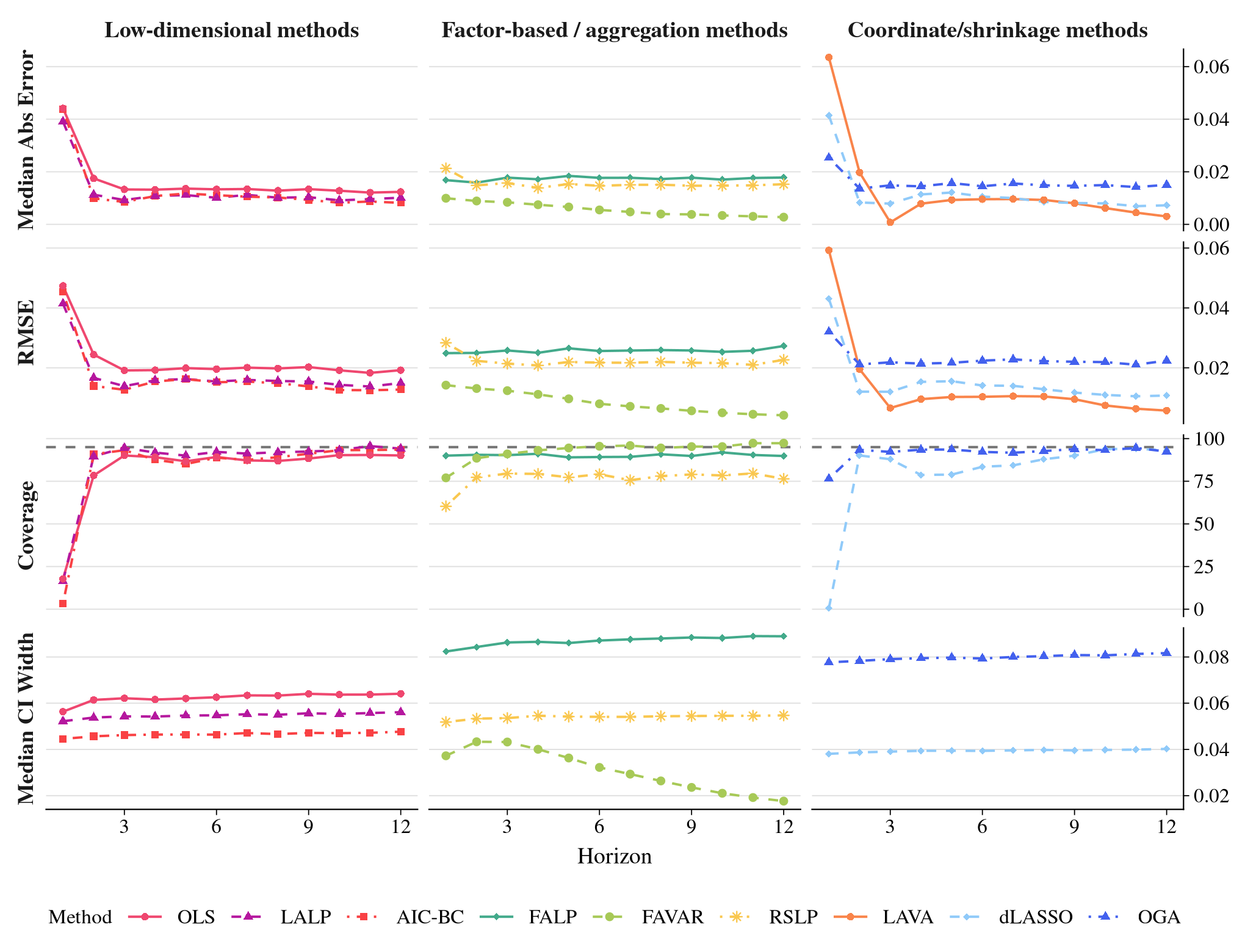}
\end{center}
        \vspace{-.4cm}
    {\footnotesize \textit{Notes.} The figure reports the median absolute error, RMSE, 95\% coverage, and median confidence-interval width across horizons. Method abbreviations are defined in Table~\ref{tab:simulation-methods}. LAVA is reported only in the point-error panels.}
\end{figure}



Figure~\ref{fig:sim_factor} reports the results. The factor structure gives the factor-based methods a clear advantage, although along different margins. FAVAR delivers the lowest point-error measures and narrowest intervals at most horizons, but its coverage is weak at the first horizon before moving close to the nominal level. FALP achieves the best first-horizon coverage and avoids the severe short-horizon undercoverage of the low-dimensional methods, although its coverage remains modestly below \(95\%\) over most horizons and its intervals are relatively wide. RSLP exhibits more pronounced undercoverage.

The proposed method does not dominate in this design. Its first-horizon point error is lower than that of the low-dimensional benchmarks, and its coverage is close to nominal at most later horizons. At the first horizon, however, coverage is weaker than for FALP, and its point-error measures generally exceed those of FAVAR thereafter. Together with Design~1, these results show that the structure of the nuisance component matters: Coordinate selection is well targeted when contamination is represented by a small number of observed variables, and factor methods are naturally advantaged when contamination is pervasive and low-rank. The two approaches should therefore be viewed as complementary.

%% file: sections/6_Empirical.tex
This section applies the proposed procedure to a monetary policy shock design in a high-dimensional macroeconomic environment. I follow the empirical setting of \citet{adamek2023_hdlp}, using the processed version of the September 2021 vintage of the FRED-MD database supplied in their replication package \citep{adamekHDLPreproReplication2024}; see also \citet{fredmd2021,mccracken2016fred}. The responses are the federal funds rate (FFR), industrial production (IP), and the CPI, and the shock variable is the FFR, following \citet{bernankeMeasuringEffectsMonetary2005}.
Identification uses the slow/fast classification in \citet{adamek2023_hdlp}. Slow-moving variables are restricted not to respond contemporaneously to an unexpected change in the FFR, whereas fast-moving variables may respond within the same month. This implements the recursive partial-identification scheme described in their empirical application.

The complete information vector contains 122 macroeconomic series: 67 slow variables, the FFR, and 54 fast variables. The local projection specification includes the contemporaneous FFR, all 67 contemporaneous slow variables, and 13 lags of all 122 series. Excluding the intercept, this yields $1,654$ 
potential regressors. The original sample contains 586 monthly observations; after allowing for 13 lags, the effective sample at horizon one is 572 observations and decreases with the projection horizon. Thus, the unrestricted LP design is genuinely high-dimensional, with \(p>T\). I use the transformations, sample period, and slow/fast classification reported in Section~4.1 and Supplement~S2 of \citet{adamek2023_hdlp}.

All the methods in this comparison use this common underlying dataset and identification scheme but differ in how they summarize or select the information set. FAVAR uses three factors and a VAR with 13 lags. FALP uses six principal components and includes 13 lags of the response, the FFR, and the estimated factors while retaining the 67 contemporaneous slow variables as always-included identification controls. Debiased LASSO follows the high-dimensional LP implementation of \citet{adamek2023_hdlp}. The proposed double OGA+HDAIC procedure selects directly from the full set of 13 lags of the 122 variables, with the contemporaneous slow variables always included. It applies OGA+HDAIC separately to the shock equation and to each horizon-specific outcome equation, and then estimates the final LP using the union of the selected controls together with the always-included controls.

Figure~\ref{fig:app_monetary} reports the estimated impulse responses. The purpose of the comparison is to document how the proposed approach behaves relative to factor-based and regularization-based approaches when applied to the same empirical design.



\begin{figure}[hbt!]
    \begin{center}
        \caption{Monetary shock results}\label{fig:app_monetary}
        \includegraphics[width=\textwidth]{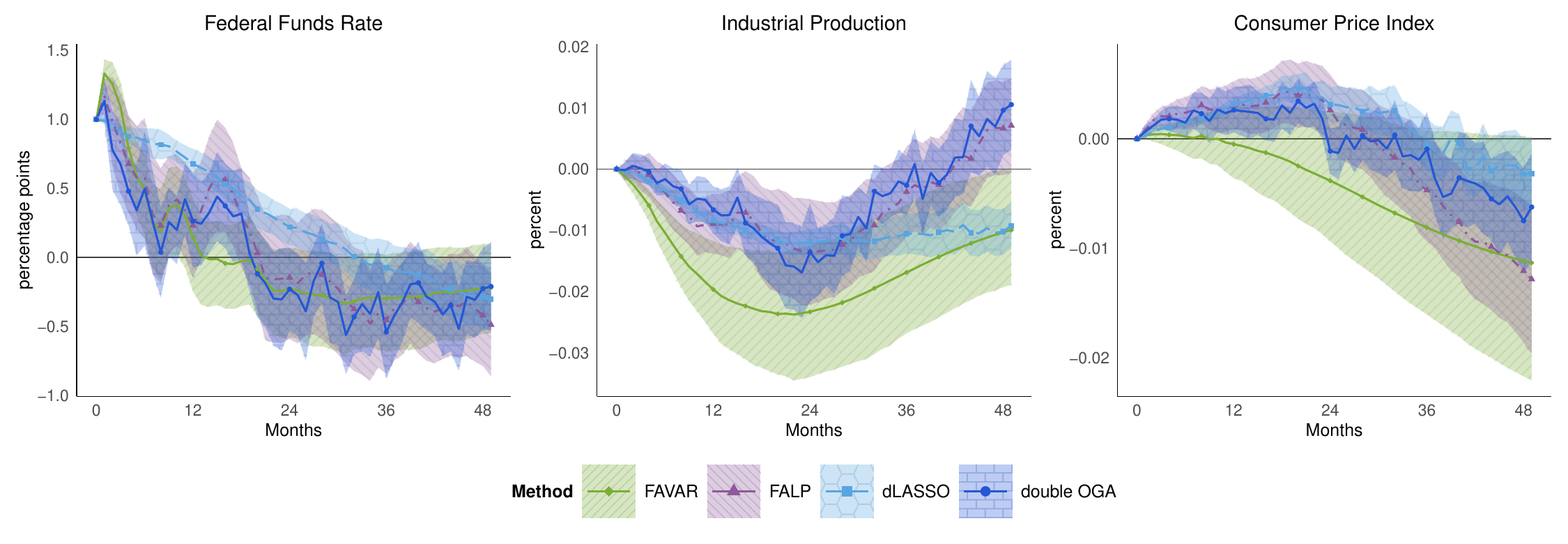}
    \end{center}
    \vspace*{-.4cm}
    {\footnotesize \textit{Notes.} The figure reports impulse responses to a monetary policy shock.
The responses are the federal funds rate, industrial production, and the consumer price index. The methods compared are FALP, FAVAR, debiased LASSO, and the proposed double OGA+HDAIC. Shaded areas report $95\%$ confidence intervals.}
\end{figure}

The estimates display a response-specific pattern. For the FFR, the response estimated with the proposed method is close to the factor-based estimates, especially FALP, and shows a gradual decline following the monetary policy shock. For IP, it is also similar to the FALP estimate over most horizons. For CPI, it lies between the factor-based estimates and the dLASSO response. Thus the proposed method does not mechanically reproduce any single benchmark, but the estimated responses remain broadly aligned with those from standard high-dimensional alternatives.

This pattern is consistent with the construction of the estimator. Factor-based methods summarize the information set through common components, while dLASSO uses a regularization-based approach to the full high-dimensional regression. The proposed method instead selects observed lagged controls separately for the shock equation and each horizon-specific outcome equation. Its response-specific behavior therefore reflects how the large information set is used across equations and horizons.

The selected controls should be interpreted as auxiliary nuisance controls rather than structural objects of interest. They adjust for predictable variation in the shock and outcome equations, and their identities need not form a smooth pattern across horizons or lags. Accordingly, the selection of an isolated lag or the appearance of a variable at only a few horizons should not be interpreted as evidence of a distinct structural channel. Variables required by the shock definition or identifying strategy, such as the contemporaneous slow variables in this application, are placed in the always-included set rather than left to selection. Online Appendix~S4.3 reports the number and identities of selected controls across responses and horizons as a transparency exercise.

%% file: sections/7_Conclusion.tex
LP provides a flexible framework for impulse response analysis, but its implementation becomes challenging when researchers condition on large information sets. In such settings, the main issue is not only dimensionality, but also how the high-dimensional controls enter the shock equation and the horizon-specific outcome equations. This paper develops a double-selection LP procedure based on OGA+HDAIC to address this problem.

The proposed method contributes in two ways. First, it separates selection for the shock equation from selection for the outcome equation, which is important when the observed shock contains predictable components related to the control set. 
Second, it accommodates nonexact sparsity by allowing nuisance coefficients to exhibit algebraic decay rather than requiring exact sparsity
with an information criterion selecting the model size along the greedy path.
This flexibility is useful in macroeconomic applications, where predictive information may be distributed across many variables and the degree of concentration is typically unknown.

The theoretical analysis adapts the OGA+HDAIC framework of \citet{ing2020} to double-selection local projections with functionally dependent data. Using functional-dependence tools, high-dimensional maximal inequalities, and a double-selection argument, I establish prediction-error bounds, \(\sqrt{T}\)-asymptotic normality for the impulse response coefficients, and consistency of HAC standard errors.

The simulations show that the value of high-dimensional controls depends on the structure of the nuisance component. Shock-equation contamination calls for selection that targets the observed shock, predictable outcome-equation nuisance variation can generate precision gains, and factor-based methods perform particularly well when contamination is generated by a pervasive factor structure. These findings suggest that selection- and factor-based approaches are complementary rather than strict substitutes.

The empirical application illustrates that the proposed procedure can be implemented in a \(p>T\) monetary policy setting with the FRED-MD information set. The resulting impulse responses are economically interpretable and are benchmarked against factor-based and regularization-based alternatives. Together, the results emphasize that high-dimensional impulse response estimation depends critically on how the information set is used to adjust both the shock and outcome equations.

%% file: sections/8_Appendix_v2.tex

\section{Auxiliary results used in the proofs}

The following auxiliary results are used repeatedly below. Proofs of Lemmas~\ref{lem:gram-score}--\ref{lem:nuisance-remainders}, together with additional technical lemmas used to prove Proposition~\ref{prop:hac-consistency}, are collected in Online Appendix~S1.

\begin{lemma}[Uniform Gram and score bounds] \label{lem:gram-score} 
Under Assumptions~\ref{asm:FD}--\ref{asm:add}, uniformly over the relevant
Gram and score processes,
\[
\max_{j,k\le p}
\left|
\frac{1}{T}\sum_{t=1}^T
\{w_{t,j}w_{t,k}-E[w_{t,j}w_{t,k}]\}
\right|
=O_p(a_T),
\]
and, for each disturbance \(\eta_t\) appearing in the auxiliary regressions,
\[
\max_{j\le p}
\left|
\frac{1}{T}\sum_{t=1}^T w_{t,j}\eta_t
\right|
=O_p(a_T).
\]
\end{lemma}

\begin{lemma}[Generic OGA+HDAIC rate]\label{lem:oga-generic} 
Consider the generic nuisance regression
\[
    z_t=w_t'\xi+\eta_t,
    \qquad
    E[w_t\eta_t]=0,
\]
where \(\xi\) satisfies Assumption~\ref{asm:poly}. 
Let \(\widehat J\) be the model selected by OGA+HDAIC and let \(\widehat\xi\) be the least-squares refit on \(W(\widehat J)\).

For any index set \(J\subseteq[p]\), define the population projection $\bar\xi(J) = \arg\min_{\operatorname{supp}(\theta)\subseteq J} E\{w_t'(\xi-\theta)\}^2$, and define the corresponding population approximation error $\sigma_\xi^2(J) = E\{w_t'(\xi-\bar\xi(J))\}^2 $.
Under Assumptions~\ref{asm:FD}--\ref{asm:add}, $|\widehat J|\le M_T^*$ by  construction, and
\[
    \frac{1}{T}
    \|W(\widehat\xi-\xi)\|_2^2
    =
    O_p(a_T^{2-1/\delta}).
\]
Moreover,
\[
    \sigma_\xi^2(\widehat J)
    =
    O_p(a_T^{2-1/\delta}).
\]
\end{lemma}

\begin{lemma}[Post-selection union-refit rates]\label{lem:nuisance-remainders} 
Let $\widetilde J=\widehat J_x\cup\widehat J_y$, where \(\widehat J_x\) and \(\widehat J_y\) are the OGA+HDAIC  selected sets from the shock equation \eqref{eq:lp_vector_shock} and the horizon-\(h\) outcome nuisance equation \eqref{eq:lp_vector_outcome}, respectively. Let \(\widetilde\pi\) and \(\widetilde\lambda\) denote the least-squares refit coefficients from regressing \(\boldsymbol x\) and \(\boldsymbol y\), respectively, on \(W(\widetilde J)\).

Under Assumptions~\ref{asm:FD}--\ref{asm:add},
\[
    |\widetilde J|
    =
    O_p(a_T^{-1/\delta}),
\]
\[
    \frac{1}{\sqrt T}
    \|W(\widetilde\pi-\pi)\|_2
    =
    O_p(a_T^{1-1/(2\delta)}),
    \qquad
    \frac{1}{\sqrt T}
    \|W(\widetilde\lambda-\lambda)\|_2
    =
    O_p(a_T^{1-1/(2\delta)}),
\]
and
\[
    \|\widetilde\pi-\pi\|_1
    =
    O_p(a_T^{1-1/\delta}),
    \qquad
    \|\widetilde\lambda-\lambda\|_1
    =
    O_p(a_T^{1-1/\delta}).
\]
Furthermore,
\[
    \sqrt T
    \left|
    (\widetilde\pi-\pi)'W'\boldsymbol{u}/T
    \right|
    =
    o_p(1),\qquad
    \sqrt T
    \left|
    (\widetilde\lambda-\lambda)'W'\boldsymbol v/T
    \right|
    =
    o_p(1),
\]
and
\[
    \sqrt T
    \left|
    (\widetilde\pi-\pi)'W'W
    (\widetilde\lambda-\lambda)/T
    \right|
    =
    o_p(1).
\]
\end{lemma}

\begin{prop}[HAC consistency]\label{prop:hac-consistency}
Under Assumptions~\ref{asm:FD}--\ref{asm:add} and the kernel and bandwidth conditions of Theorem~\ref{thm:variance}, we have
\[
    \widehat\Omega_h-\Omega_h=o_p(1),
    \qquad
    \widehat\tau^2-\tau^2=o_p(1).
\]
\end{prop}

\begin{proof}
    Let \(\widetilde\Omega_h\) denote the infeasible counterpart of
\(\widehat\Omega_h\), obtained by replacing
\(\widehat\psi_{t,h}=\widehat v_t\widehat u_{t,h}\) with
\(\psi_{t,h}=v_tu_{t,h}\). That is,
\[
    \widetilde\Omega_h
    =
    \widetilde\Gamma_{0,h}
    +
    2\sum_{\ell=1}^{L_h}
    K(\ell/L_h)\widetilde\Gamma_{\ell,h},
    \qquad
    \widetilde\Gamma_{\ell,h}
    =
    \frac1T\sum_{t=\ell+1}^T
    \psi_{t,h}\psi_{t-\ell,h}.
\]
Decompose
\[
    \widehat\Omega_h-\Omega_h
    =
    \bigl(\widehat\Omega_h-\widetilde\Omega_h\bigr)
    +
    \bigl(\widetilde\Omega_h-\Omega_h\bigr).
\]
The two terms are \(o_p(1)\) by Online Appendix Lemmas~S12 and~S9, respectively. The second claim follows from Online Appendix Lemma~S7.
\end{proof}
\section{Proof of Theorem~\ref{thm:errbd}}\label{sec:Appendix:Thrm:errbd}
\begin{proof}
The result follows directly from Lemma~\ref{lem:oga-generic}. Applying that
lemma to the generic nuisance regression
\[
    z_t=w_t'\xi+\eta_t,
    \qquad
    E[w_t\eta_t]=0,
\]
with the HDAIC penalty calibrated as \(\kappa_T=Ta_T^2\), gives
\[
    \frac{1}{T}
    \|W(\widehat\xi-\xi)\|_2^2
    =
    O_p(a_T^{2-1/\delta}).
\]
This proves the theorem.
\end{proof}

\section{Proof of Theorem~\ref{thm:single}}\label{sec:Appendix:Thrm:single}
\begin{proof}
Consider the models defined in \eqref{eq:lp_vector_structural}--\eqref{eq:lp_vector_shock}.
Let $\widetilde J=\widehat J_x\cup\widehat J_y$, and let \(M_{\widetilde J}\) be the residual-maker matrix for \(W(\widetilde J)\). Define the union-refit nuisance estimators \(\widetilde\pi\) and \(\widetilde\lambda\), and set $d_\pi=\widetilde\pi-\pi$ and $d_\lambda=\widetilde\lambda-\lambda$.
By the Frisch--Waugh--Lovell theorem,
\[
    \widehat\beta_h-\beta_h
    =
    \left(
    \frac{\boldsymbol{x}'M_{\widetilde J}\boldsymbol{x}}{T}
    \right)^{-1}
    \left(
    \frac{\boldsymbol{x}'M_{\widetilde J}(\boldsymbol y-\beta_h \boldsymbol x)}{T}
    \right).
\]
Let 
\begin{align}\label{eq:def:vhat}
\widehat{\boldsymbol{v}} = M_{\widetilde J}\boldsymbol{x} = \boldsymbol{x}-W\widetilde\pi = \boldsymbol v-Wd_\pi.
\end{align}
By \eqref{eq:def:vhat}, we have $\boldsymbol{x}'M_{\widetilde J}\boldsymbol{x}=\widehat{\boldsymbol{v}}'\widehat{\boldsymbol{v}}$.
Therefore
\[
    \sqrt T(\widehat\beta_h-\beta_h)
    =
    A_T^{-1}B_T,
\]
where
\[
    A_T=\frac{\widehat{\boldsymbol{v}}'\widehat{\boldsymbol{v}}}{T},
    \qquad
    B_T=\frac{\widehat{\boldsymbol{v}}'(\boldsymbol y-\beta_h \boldsymbol x)}{\sqrt T}.
\]

We first show
\[
    A_T=\frac{\boldsymbol{v}'\boldsymbol{v}}{T}+o_p(1).
\]
Using \eqref{eq:def:vhat},
\[
\begin{aligned}
    A_T
    &=
    \frac{\boldsymbol{v}'\boldsymbol{v}}{T}
    -
    2\frac{d_\pi'W'\boldsymbol{v}}{T}
    +
    \frac{d_\pi'W'Wd_\pi}{T}.
\end{aligned}
\]
By Lemma~\ref{lem:nuisance-remainders},
\[
    \|d_\pi\|_1=O_p(a_T^{1-1/\delta}),
    \qquad
    \frac{1}{\sqrt T}\|Wd_\pi\|_2
    =
    O_p(a_T^{1-1/(2\delta)}).
\]
Also, by Lemma~\ref{lem:gram-score},
\[
    \|W'\boldsymbol{v}/T\|_\infty=O_p(a_T).
\]
Hence
\[
    \left|
    \frac{d_\pi'W'\boldsymbol{v}}{T}
    \right|
    \le
    \|d_\pi\|_1\|W'\boldsymbol{v}/T\|_\infty
    =
    O_p(a_T^{2-1/\delta})
    =
    o_p(1),
\]
and
\[
    \frac{d_\pi'W'Wd_\pi}{T}
    =
    \frac{1}{T}\|Wd_\pi\|_2^2
    =
    O_p(a_T^{2-1/\delta})
    =
    o_p(1).
\]
Therefore,
\[
    A_T=\frac{\boldsymbol{v}'\boldsymbol{v}}{T}+o_p(1).
\]

Next consider \(B_T\). 
Using the outcome and shock nuisance equations in \eqref{eq:lp_vector_outcome}--\eqref{eq:lp_vector_shock} together with $\boldsymbol e=\beta_h\boldsymbol v+\boldsymbol u$ and since \(\widehat{\boldsymbol{v}}\) is orthogonal to \(W(\widetilde J)\), we may add and subtract \(W\widetilde\lambda\) and \(W\widetilde\pi\). Thus
\[
\begin{aligned}
    B_T
    &=
    \frac{\widehat{\boldsymbol{v}}'(\boldsymbol y-\beta_h \boldsymbol x)}{\sqrt T} \\
    &=
    \frac{\widehat{\boldsymbol{v}}'
    \{W(\lambda-\widetilde\lambda)+\boldsymbol{e}
    -\beta_h W(\pi-\widetilde\pi)-\beta_h \boldsymbol{v}\}}
    {\sqrt T} \\
    &=
    \frac{\widehat{\boldsymbol{v}}'\boldsymbol{u}}{\sqrt T}
    -
    \frac{\widehat{\boldsymbol{v}}'Wd_\lambda}{\sqrt T}
    +
    \beta_h\frac{\widehat{\boldsymbol{v}}'Wd_\pi}{\sqrt T}.
\end{aligned}
\]
Using \(\widehat{\boldsymbol{v}}=\boldsymbol{v}-Wd_\pi\), the first term satisfies
\[
    \frac{\widehat{\boldsymbol{v}}'\boldsymbol{u}}{\sqrt T}
    =
    \frac{\boldsymbol{v}'\boldsymbol{u}}{\sqrt T}
    -
    \sqrt T\frac{d_\pi'W'\boldsymbol{u}}{T}.
\]
By Lemma~\ref{lem:nuisance-remainders},
\[
    \sqrt T
    \left|
    \frac{d_\pi'W'\boldsymbol{u}}{T}
    \right|
    =
    o_p(1).
\]
Therefore
\[
    \frac{\widehat{\boldsymbol{v}}'\boldsymbol{u}}{\sqrt T}
    =
    \frac{\boldsymbol{v}'\boldsymbol{u}}{\sqrt T}
    +
    o_p(1).
\]

For the second term,
\[
    \frac{\widehat{\boldsymbol{v}}'Wd_\lambda}{\sqrt T}
    =
    \sqrt T\frac{\boldsymbol{v}'Wd_\lambda}{T}
    -
    \sqrt T\frac{d_\pi'W'Wd_\lambda}{T}.
\]
By Lemma~\ref{lem:nuisance-remainders},
\[
    \sqrt T
    \left|
    \frac{d_\lambda'W'\boldsymbol{v}}{T}
    \right|
    =
    o_p(1), \quad \text{and} \quad
    \sqrt T
    \left|
    \frac{d_\pi'W'Wd_\lambda}{T}
    \right|
    =
    o_p(1).
\]
Hence
\[
    \frac{\widehat{\boldsymbol{v}}'Wd_\lambda}{\sqrt T}
    =
    o_p(1).
\]

For the third term,
\[
    \frac{\widehat{\boldsymbol{v}}'Wd_\pi}{\sqrt T}
    =
    \sqrt T\frac{\boldsymbol{v}'Wd_\pi}{T}
    -
    \sqrt T\frac{d_\pi'W'Wd_\pi}{T}.
\]
The first component is \(o_p(1)\) by
\[
    \sqrt T
    \left|
    \frac{d_\pi'W'\boldsymbol{v}}{T}
    \right|
    \le
    \sqrt T
    \|d_\pi\|_1
    \|W'\boldsymbol{v}/T\|_\infty
    =
    O_p(\sqrt T a_T^{2-1/\delta})
    =
    o_p(1),
\]
under Assumption~\ref{asm:add}\ref{asm:add:rate}. The second component is
\[
    \sqrt T
    \frac{d_\pi'W'Wd_\pi}{T}
    =
    \sqrt T
    O_p(a_T^{2-1/\delta})
    =
    o_p(1).
\]
Therefore
\[
    \frac{\widehat{\boldsymbol{v}}'Wd_\pi}{\sqrt T}=o_p(1).
\]
Combining the preceding displays gives
\[
    B_T
    =
    \frac{\boldsymbol{v}'\boldsymbol{u}}{\sqrt T}
    +
    o_p(1).
\]

By Online Appendix Lemma~S1 
with \(k=2\), zero lags, and inputs \((v,v)\), the process \(v_t^2-E[v_t^2]\) has finite \(q/2\)-moment and summable physical-dependence coefficients. Hence, by the law of large numbers for functionally dependent processes,
\[
    \frac{\boldsymbol{v}'\boldsymbol{v}}{T}
    =
    \frac{1}{T}\sum_{t=1}^T v_t^2
    \overset{p}{\to}
    \tau^2>0.
\]
Since $A_T=\frac{\boldsymbol{v}'\boldsymbol{v}}{T}+o_p(1)$, it follows that $A_T\overset{p}{\to}\tau^2$ and
\[
    A_T^{-1}
    =
    \left(
    \frac{\boldsymbol{v}'\boldsymbol{v}}{T}
    \right)^{-1}
    +
    o_p(1).
\]

Let $\psi_{t,h}=v_tu_{t,h}$. By Online Appendix Lemma~S1 with \(k=2\), zero lags, and inputs \((v,u_h)\),
\[
    \|\psi_{t,h}\|_{q/2}\le C,
    \qquad
    \sum_{j=0}^{\infty}\delta_{j,q/2}(\psi_h)<\infty .
\]
Since \(q/2>2\), Lyapunov's inequality gives $\delta_{j,2}(\psi_h) \le \delta_{j,q/2}(\psi_h)$ for every \(j\ge0\). Hence
\[
    \sum_{j=0}^{\infty}\delta_{j,2}(\psi_h)<\infty .
\]
Therefore the central limit theorem for functionally dependent processes \citep[e.g.][]{wuNonlinearSystemTheory2005} gives
\[
    \frac{1}{\sqrt T}\sum_{t=1}^T v_tu_{t,h}
    \Rightarrow
    N(0,\Omega_h).
\]
In particular,
\[
    \frac{\boldsymbol{v}'\boldsymbol{u}}{\sqrt T}=O_p(1).
\]

Combining $B_T = \frac{\boldsymbol{v}'\boldsymbol{u}}{\sqrt T} + o_p(1)$ with the preceding display for \(A_T^{-1}\), we obtain
\[
    \sqrt T(\widehat\beta_h-\beta_h)
    =
    \left(
    \frac{\boldsymbol{v}'\boldsymbol{v}}{T}
    \right)^{-1}
    \left(
    \frac{\boldsymbol{v}'\boldsymbol{u}}{\sqrt T}
    \right)
    +
    o_p(1).
\]
Slutsky's theorem then yields
\[
    \sqrt T(\widehat\beta_h-\beta_h)
    \Rightarrow
    N(0,\Omega_h/\tau^4).
\]
This proves the theorem.
\end{proof}

\section{Proof of Corollary~\ref{cor:single}}\label{sec:Appendix:Thrm:cor}
\begin{proof}
Using the notation of Corollary~\ref{cor:single}, stacking \eqref{eq:dt_structural}--\eqref{eq:dt_xonw} over \(t\) and premultiplying the resulting equations by \(M_D\) gives
\[
    \widetilde{\boldsymbol{y}}
    =
    \beta_h\widetilde{\boldsymbol{x}}+\widetilde W\gamma_h+\widetilde{\boldsymbol{u}},
    \qquad
    \widetilde{\boldsymbol{x}}
    =
    \widetilde W\pi+\widetilde{\boldsymbol{v}},
\]
and, in reduced-form notation,
\[
    \widetilde{\boldsymbol{y}}
    =
    \widetilde W\lambda+\widetilde{\boldsymbol{e}},
    \qquad
    \widetilde{\boldsymbol{e}}
    =
    \beta_h\widetilde{\boldsymbol{v}}+\widetilde{\boldsymbol{u}},
\]
where
\[
    \widetilde{\boldsymbol{u}}=M_D\boldsymbol{u},
    \qquad
    \widetilde{\boldsymbol{v}}=M_D\boldsymbol{v},
    \qquad
    \widetilde{\boldsymbol{e}}=M_D\boldsymbol{e} .
\]

For any subset \(J\subseteq\{1,\ldots,p\}\), let \(W(J)\) denote the submatrix of \(W\) and let \(\widetilde W(J)=M_DW(J)\). By the Frisch--Waugh--Lovell theorem, the residual from regressing any vector \(z\) on \((D,W(J))\) is identical to the residual from regressing \(M_Dz\) on \(\widetilde W(J)\). That is,
\[
    M_{[D,W(J)]}z
    =
    M_{\widetilde W(J)}M_Dz .
\]
Therefore applying Algorithm~\ref{alg:double} to
\[
    (\widetilde{\boldsymbol{y}},\widetilde{\boldsymbol{x}},\widetilde W)
\]
is equivalent to applying the original procedure with \(d_t\) included
unpenalized in every selection and final-estimation regression.

By Lemma~S6, the sample-residualized Gram and score arrays satisfy the same bounds used in the proof of Theorem~\ref{thm:single}. Repeating that proof for
\((\widetilde{\boldsymbol{y}},\widetilde{\boldsymbol{x}},\widetilde W)\) therefore gives
\[
    \sqrt T(\widehat\beta_h-\beta_h)
    =
    \left(
        \frac{\widetilde{\boldsymbol{v}}'\widetilde{\boldsymbol{v}}}{T}
    \right)^{-1}
    \left(
        \frac{\widetilde{\boldsymbol{v}}'\widetilde{\boldsymbol{u}}}{\sqrt T}
    \right)
    +o_p(1).
\]
By Lemma~S6,
\[
    \frac{\widetilde{\boldsymbol{v}}'\widetilde{\boldsymbol{v}}}{T}
    =
    \frac{\boldsymbol{v}'\boldsymbol{v}}{T}+o_p(1),
    \qquad
    \frac{\widetilde{\boldsymbol{v}}'\widetilde{\boldsymbol{u}}}{\sqrt T}
    =
    \frac{\boldsymbol{v}'\boldsymbol{u}}{\sqrt T}+o_p(1).
\]
The conclusion follows from the same CLT and Slutsky argument used in
Theorem~\ref{thm:single}.

\end{proof}

\section{Proof of Theorem~\ref{thm:variance}}\label{sec:Appendix:Thrm:variance}
\begin{proof}
    By Proposition~\ref{prop:hac-consistency},
\[
    \widehat\Omega_h\overset{p}{\to}\Omega_h,
    \qquad
    \widehat\tau^2\overset{p}{\to}\tau^2.
\]
Since \(\tau^2>0\), the continuous mapping theorem gives
\[
    \widehat\tau^{-4}\overset{p}{\to}\tau^{-4}.
\]
Therefore,
\[
    \widehat\sigma_h^2
    =
    \widehat\Omega_h/\widehat\tau^4
    \overset{p}{\to}
    \Omega_h/\tau^4
    =
    \sigma_h^2 .
\]
\end{proof}